\documentclass[runningheads]{llncs}

\usepackage{algorithm}
\usepackage{algpseudocode}
\usepackage{amsmath}
\usepackage{amsfonts}
\usepackage{amssymb}
\usepackage{color}
\usepackage{comment}
\usepackage{enumitem}
\usepackage{graphicx}
\usepackage[bookmarks=false]{hyperref}
\usepackage[utf8]{inputenc}
\usepackage{listings}
\usepackage{makeidx}  
\usepackage[square,sort,comma,numbers]{natbib}
\usepackage{subcaption}
\usepackage{url}
\usepackage{wrapfig}
\usepackage{xcolor}
\usepackage[pdf]{graphviz}
\usepackage{multirow}
\usepackage{refcount}

\usepackage[graphicx]{realboxes}
\usepackage{tikz}
\tikzset{elliptic state/.style={draw,ellipse}}
\usetikzlibrary{shapes,shapes.geometric,arrows,fit,calc,positioning,automata,}
\usetikzlibrary{automata,positioning}
\usepackage{tikz-qtree}
\usepackage{float}
\usepackage{floatflt}

\newcommand*{\tool}{\textsc{SGR}($k$)}

\newcommand*{\F}{\mathcal{F}}

\newcommand*{\Statess}{\mathit{S}}

\newcommand*{\Buchi}{B\"uchi }
\newcommand*{\GS}{\mathit{GS}} 

\newcommand*{\I}{\mathcal{I}}
\renewcommand*{\O}{\mathcal{O}}
\newcommand{\J}{\mathcal J}
\newcommand{\V}{\mathcal V}
\newcommand{\U}{\mathcal U}
\newcommand{\inp}{{\mathit{i}}}
\newcommand{\out}{\mathit{o}}

\newcommand{\travel}{{\mathit{travel}}}

\newcommand{\gar}{{\mathit{gar}}}
\newcommand{\acc}{{\mathit{acc}}}

\newcommand{\Adapter}{{\mathit{Adapter}}}

\newcommand{\Adaptee}{{\mathit{Adaptee}}}
\newcommand{\Target}{{\mathit{Target}}}
\newcommand{\Input}{{\mathit{Input}}}
\newcommand{\Output}{{\mathit{Output}}}

\newcommand{\ltl}{\mathsf{LTL}}
\newcommand{\ltlU}{\mathsf{U}}
\newcommand{\ltlX}{\mathsf{X}}
\newcommand{\ltlNeg}{\neg}
\newcommand{\ltlG}{\mathsf{G}}
\newcommand{\ltlF}{\mathsf{F}}

\newcommand{\SCC}{\mathsf{\mathsf{SCC}}}
\newcommand{\Reach}{\mathsf{\mathsf{Reach}}}
\newcommand{\Reachinverse}{\mathsf{\mathsf{Reach}^{-1}}}

\newcommand{\terminal}{\mathsf{\mathsf{Terminal}}}
\newcommand{\C}{\mathsf{\mathsf{DC}}}
\newcommand{\Acc}{\mathsf{\mathsf{Acc}}}
\newcommand{\NC}{\mathsf{\mathsf{N}}}
\newcommand{\AC}{\mathsf{\mathsf{A}}}

\newcommand{\winning}{\mathsf{\mathsf{Win}}}
\newcommand{\winningAux}{\mathsf{\mathsf{Win\_Aux}}}
\newcommand{\strategy}{\mathsf{\mathsf{Fb}}}
\newcommand{\strategyAux}{\mathsf{\mathsf{Fb\_Aux}}}
\newcommand{\Reachability}{\mathsf{\mathsf{Reachability}}}
\newcommand{\FReachability}{\mathsf{\mathsf{FReachability}}}
\newcommand{\Safety}{\mathsf{\mathsf{Safety}}}
\newcommand{\FSafety}{\mathsf{\mathsf{FSafety}}}
\newcommand{\InitIn}{\mathsf{\mathsf{InitIn}}}
\newcommand{\InitOut}{\mathsf{\mathsf{InitOut}}}

\newcommand{\X}{\mathit{X}}

\usepackage[colorinlistoftodos,prependcaption,textsize=tiny]{todonotes}

\begin{document}
\title{Adapting Behaviors via Reactive Synthesis}

%
%
\author{Gal Amram\inst{1} \and
Suguman Bansal\inst{2} \and
Dror Fried\inst{3} \and
Lucas M. Tabajara\inst{4} \and
Moshe Y. Vardi\inst{4} \and
Gera Weiss\inst{5}}

%
\authorrunning{G. Amram et al.}
\institute{Tel-Aviv University, Israel \email{galam1483@gmail.com} \and
University of Pennsylvania, USA \email{suguman@seas.upenn.edu} \and
The Open University of Israel, Israel \email{dfried@openu.ac.il} \and
Rice University, USA \email{\{lucasmt,vardi\}@rice.edu} \and
Ben-Gurion University of the Negev, Israel \email{geraw@bgu.ac.il}
}

%
\maketitle              

\begin{abstract}
In the \emph{Adapter Design Pattern}, a programmer implements a \emph{Target} interface by constructing an \emph{Adapter} that accesses an existing \emph{Adaptee} code.  In this work, we present a reactive synthesis interpretation to the adapter design pattern, wherein an algorithm takes an \emph{Adaptee} and a \emph{Target} transducers, and the aim is to synthesize an \emph{Adapter} transducer that, when composed with the {\em Adaptee}, generates a behavior that is equivalent to the behavior of the {\em Target}. One use of such an algorithm is to synthesize controllers that achieve similar goals on different hardware platforms.  While this problem can be solved with existing synthesis algorithms, current state-of-the-art tools fail to scale.  
To cope with the computational complexity of the problem, we introduce a special form of specification format, called {\em Separated GR($k$)}, which can be solved with a scalable synthesis algorithm but still allows for a large set of realistic specifications. 
We solve the realizability and the synthesis problems for Separated GR($k$), and show how to exploit the separated nature of our specification to construct better algorithms, in terms of time complexity, than known algorithms for GR($k$) synthesis. We then describe a tool, called SGR($k$), that we have implemented based on the above approach and show, by experimental evaluation, how our tool outperforms current state-of-the-art tools on various benchmarks and test-cases.
\end{abstract}

\sloppy

\section{Introduction}
\label{sec:intro}
Inspired by the well known adapter design pattern~\cite{gamma1995design}, we study the use of reactive synthesis for generating adapters that translate inputs meant for a target transducer to inputs of an adaptee transducer. Consider, as one motivating example, the practice of adding code to an operating system that mitigates the risk posed by newly discovered hardware vulnerabilities like Spectre and Meltdown~\cite{Kocher2018spectre, Lipp2018meltdown}.
While the discovery of such vulnerabilities puts constraints on how the hardware can be used, the patch of the operating system (called adapter in this paper) takes upon itself to take care of running all applications without change~\cite{220586}. It does so by allowing applications of the existing interface, while adapting their operation in way that ensures that the system is not exposed to the new threat.

Formally, we propose the following synthesis problem: given two finite-state transducers called $\Target$ and $\Adaptee$, synthesize a finite-state transducer called $Adapter$ such that
\[    \Adaptee \circ \Adapter \eqsim \Target.\]
The symbol $\circ$ stands for standard transducer composition~\cite{lothaire2005applied} and the symbol $\eqsim$  stands for an equivalence relation, a generalization of sequential equality, which we explain below. In words, we want an $\Adapter$ that stands between an $\Adaptee$ and its inputs and guarantees, such that the composition $\Adaptee \circ \Adapter$ is equivalent to $\Target$. 
In the vulnerability patching example, $\Adaptee$ is a model of the 
constrained hardware and $\Target$ is a model of the hardware as used before the discovery of the vulnerability, without the new constraints. The $\Adapter$ that we generate models the patch that mediates between the vulnerable hardware and applications that are not aware of the vulnerability.

In our setting, an input to the synthesis algorithm is the equivalence relation  along with the specification of the adaptee and of the target.
While the problem of synthesizing an adapter such that $\Adaptee \circ \Adapter$ is sequentially equal to $\Target$ may be useful in some cases~\cite{10.5555/972695.972698}, we study here a more general problem. This is called for by applications such as the vulnerability covering patches described above. Specifically, we allow our users to specify an equivalence relation between $\Adaptee \circ \Adapter$ and $\Target$  that is not necessarily sequential equality.  In this paper, we propose to use $\omega$-regular properties~\cite{de2000concurrent} for specifying this equivalence relation, as follows. We assume, without loss of generality, that the outputs of both the $\Target$ and the $\Adaptee$ are assignments to disjoint sets of atomic propositions. We then consider sequences of pairs of such assignments that correspond to zipped runs of $\Adaptee \circ \Adapter$ and of $\Target$ over the same input. Having this set of sequences in mind, the user specifies a set of temporal  properties using an $\omega$-regular formalism such as LTL or B\"uchi automata. The transducer $\Adaptee \circ \Adapter$ is considered equivalent to $\Target$  if all the properties that the user specified are satisfied for  each sequence in the set~\cite{6280452}.  Note that the equivalence relation can be very different than sequential equality, it can, for example, say that  $\Adaptee \circ \Adapter$ must be, in a way,  a ``mirror image" of $\Target$, as demonstrated by the cleaning robots example in Section~\ref{sec:illustrated}, where $\Target$ is a robot that cleans some rooms and $\Adaptee \circ \Adapter$ is a robot that clean all the rooms that $\Target$ did not clean. 

The solution that we propose in this paper consists of two phases: we first transform the transducers to transition systems and arrive at a game structure that is more amenable for game-based techniques. Then we make use of the specific form of the resulting game and some simplifying assumptions about the form of the equivalence properties to solve the game efficiently. The game structures that we analyze consist of pairs of transition systems called  $\Input$ and $\Output$,  accompanied by a set of $\omega$-regular properties that specify equivalence relation between the two, as described above. The game that we solve is, then, to find a controller that reads the assignments to the variables of the $\Input$ and produces a valid sequence of assignments to the variables of the $\Output$ such that all the properties are satisfied. The translation of the transducers to this game structure is rather direct, as elaborated in Section~\ref{sec:transducers}. The $\Input$ transition system is generated from the $\Target$ transducer and the $\Output$ transition system is generated from the $\Adaptee$ transducer. This is because we want the $\Adapter$, which we generate from the controller as described below, to consider the behavior of the $\Target$ and to translate it to a command that generates an equivalent behaviour of $\Adaptee$. Once we find a controller that solves the game, we can transform it to an adapter we detail in Section~\ref{sec:transducers}.

The synthesis problem that we defined so far is as hard computationally as general LTL synthesis and is thus double exponential in the worst case~\cite{pnueli1989synthesis}. 
To cope with this difficulty, we propose to use a well known fragment of LTL called GR($k$). GR($k$) generalizes the GR($1$) subset of LTL~\cite{BloemJPPS12}, a practical fragment  of LTL for which a feasible reactive synthesis algorithm exists (see, e.g.,~\cite{Kress-GazitFP09, MaozR16a,OzayTMW11,OzayTM11,RyzhykW16,BloemGJPPW07}).
Furthermore, GR($k$) formulas are known to be highly expressive, as they can encode most commonly appearing $\ltl$ industrial patterns~\cite{DwyerAC99,MaozR16,MenghiTPGB19} and DBA properties (see related works for details).
In addition to using GR($k$), since the $\Input$ and $\Output$ in our model are separated transition systems, with separated sets of atomic propositions, we focus on properties that separate input and output variables. 
That is, our specification has the form $\bigwedge_{i=1}^k(\phi_i\rightarrow\psi_i)$, where the $\phi_i$ and $\psi_i$ are conjunctions of LTL $\ltlG \ltlF$ (Globally in the Future) formulas over $\Input$  variables only and $\Output$ variables only respectively. 
We call this model \textit{Separated GR($k$)}. We show through several case-studies that this fragment of LTL suffices to specify a range of useful equivalence relations. 

We study the problems of realizability and synthesis on Separated GR($k$) game.
For that, we first consider a sub-problem of solving a \emph{weak B\"uchi} game.  Then we identify and make use of a property of separated games that we call \emph{delay property}:
the system can delay its response to the environment indefinitely as long as it remains in the same connected component of the game graph.
This allows us to decide the realizability of Separated GR($k$) in $O(|\varphi|+N)$ symbolic operations, and to synthesize a controller for a realizable specification in $O(|\varphi|N)$ symbolic operations, where $\varphi$ is the Separated GR($k$) specification, and $N$ is the size of the state-space. Thus, Separated GR($k$) games are easier to solve that solving GR($k$) games which  require  $O(N^{k+1}k!)$ operations~\cite{PitermanP06}. This demonstrates the efficiency of our framework, since  $|\varphi|$ tends to be smaller than $N$ and in most practical cases, $|\varphi|\in O(log(N))$.

The benefits of the complexity-theoretic improvement are reflected in empirical evaluations on our case studies of separated GR($k$) formulas. We demonstrate that while separated GR($k$) formulas are challenging for state-of-the-art synthesis tools, a symbolic BDD-based implementation of our algorithm solves them scalably and efficiently.


The rest of the paper is organized as follows:
Section~\ref{Sec:Prelims} introduces necessary preliminaries.
Separated GR($k$) games are introduced and formulated in Section~\ref{sec:problem-def}. In Section~\ref{sec:transducers} we describe how to use Separated GR($k$) games synthesis to generate the adapter transducer, and introduce several use-cases.
Next, we turn to solving separated GR($k$) games. An overview of our solution approach and a necessary property for correctness of algorithm, called the delay property, is given in Section~\ref{sec:reduction}. A complete symbolic algorithm is presented in Section~\ref{sec:algorithmsWeakBuchi}. An empirical evaluation on case-studies is presented in Section~\ref{sec:experiments}. Finally, in Section~\ref{sec:related} and Section~\ref{sec:Discussion} respectively,  we give related work and conclude.

\section{Preliminaries}
\label{Sec:Prelims}

\subsubsection{General Definitions.}
Given a set of Boolean variables $\V$, a {\em state over $\V$} is an assignment $s$ to the variables in $\V$. We describe $s$ as the subset of $\V$ that is assigned $\mathsf{True}$ in $s$. The set of {\em primed variables of $\V$} is $\V' = \{v' \mid v \in \V\}$. Then $s'=\{v'\mid v\in s\} $ is the primed state $s'$ over $\V'$.
An {\em assertion over $\V$} is a Boolean formula over variables $\V$.
A state $s$ satisfies an assertion $\rho$ over the same variables, denoted $s \models \rho$, if $\rho$ evaluates to $\mathsf{True}$ by assigning $\mathit{true}$ to the elements of $s$.
We define the \emph{projection} of a state $s$ on a subset $\U \subseteq \V$ as denoted by $s|_\U = s\cap \U$. We extend the notion of projection to a set of states $S \subseteq 2^\V$ by defining $S|_\U = \{s|_\U \mid s \in S\}$.


 Our specification is a special form of \textit{Linear Temporal Logic ($\ltl$)}. $\ltl$~\cite{Pnueli77} extends propositional logic with infinite-horizon temporal operators. 
The syntax of an $\ltl$ formula over a finite set of Boolean variables  $\V$ is  defined as follows: 
$\varphi::= v\in\V \mid \ltlNeg \varphi \mid \varphi \land \varphi \mid \varphi \lor \varphi \mid \ltlX \varphi \mid \varphi \ltlU \varphi \mid \ltlF \varphi \mid \ltlG \varphi$.
Here $\ltlX$ (Next), $\ltlU$ (Until), $\ltlF$ (Eventually), $\ltlG$ (Always) are temporal operators.  
The semantics of $\ltl$ can be found in~\cite[Chapter~5]{MCBook}.

We model the adapters as transducers. A \textit{transducer} is a deterministic finite-state machine with no accepting states, but with additional output alphabet and an additional function from the set of states to the output alphabet. A formal definition of a transducer can be found in~\cite{fisman2009reasoning}, but is not required for this paper.

The algorithms developed in this paper are symbolic,  i.e. manipulate implicit representations of sets of states. To this end, we use \textit{Binary Decision Diagrams (BDDs)}~\cite{Akers78,Bryant86} to represent assertions. For a BDD $\mathsf{B}$ and sets of variables $\mathcal V_1,\cdots \mathcal V_n$, we write $\mathsf{B}(\mathcal V_1,\dots,\mathcal V_n)$ to denote that $\mathsf{B}$ represents an assertion over $\mathcal V_1\cup\cdots\cup \mathcal V_n$. For a state $s$ over $\mathcal V$, we write $s\models B(\mathcal V)$ to denote that the assertion that $\mathsf{B}$ represents is satisfied by the state $s$. BDDs support several \textit{symbolic operations}: conjunction $(\vee)$, disjunction $(\wedge)$, negation $(\neg)$, and extraction of variables using the $\exists$ and $\forall$ operators. We measure time complexity of a symbolic algorithm by a worst case \#symbolic-operations it performs. A discussion on a rigorous treatment of BDD operations can be found in Appendix~\ref{app:prelims}.

\subsubsection{Game Structures and Games.} 
We follow the notations of~\cite{BloemJPPS12}. A game structure $\GS=(\I,\O,\theta_\I,\theta_\O,\rho_\I,\rho_\O)$ defines a turn-based interaction between an {\em environment} and a {\em system} players.
The input variables $\I$ and output variables $\O$ are two disjoint sets of Boolean variables that are controlled by the environment and system, respectively. The environment's {\em  initial assumption} $\theta_\I$ is an assertion over $\I$, and the system's \emph{initial guarantee} $\theta_\O$ is an assertion over $\I{\cup}\O$.
The environment's \emph{safety assumption} $\rho_\I$ is an assertion over $\I{\cup}\O{\cup}\I'$, where the interpretation of  $(\inp_0,\out_0,\inp_1')\models \rho_\I$ is that from state $(\inp_0,\out_0)$ the environment can assign $\inp_1$ to the input variables. W.l.o.g, we assume that $\rho_\I$ is deadlock free, i.e., for all $(\inp_0, \out_0)$ there exists an $\inp_1$ s.t. $(\inp_0,\out_0,\inp_1')\models \rho_\I$. Similarly, the system's \emph{safety guarantee} $\rho_\O$ is an assertion over $\I{\cup}\O{\cup}\I'{\cup}\O'$, where the interpretation of    
     $(\inp_0,\out_0,\inp_1', \out_1')\models \rho_\O$ is that from state $(\inp_0,\out_0)$ when the environment assigns $\inp_1$ to the input variables, the system can assign $\out_1$ to the output variables. Again, w.l.o.g, we assume that $\rho_\O$ is deadlock free, i.e., for all $(\inp_0, \out_0, \inp_1')$ there exists an $\out_1$ s.t. $(\inp_0,\out_0,\inp_1',\out_1')\models \rho_\O$.

    
    
    
    
    


A play over $\GS$ progresses by the players taking turns to assign values to their own variables ad infinitum, where the players must satisfy the initial conditions at the start and the safety conditions thereafter.
Formally, a \emph{play} $\pi=s_0,s_1,\dots$ is an infinite sequence of states over $\I\cup\O$ such that $s_0\models \theta_\I\wedge \theta_\O$ and $(s_j,s_{j+1}')\models \rho_\I\wedge \rho_\O$  for all $j\geq 0$.
A \emph{play prefix} is either a play or a finite sequence of states that can be extended to a play.
Then a {\em strategy} is a function   $f:(2^{\I\cup\O})^+\times 2^\I\rightarrow 2^\O$ such that if $s_0,\dots,s_m$ is a play prefix, $(s_m,\inp')\models \rho_\I$ and $f(s_0,\dots,s_m,\inp)=\out$, then $(s_m,\inp',\out')\models \rho_\O$.
Intuitively, a strategy directs the system on what to assign to the output variables, depending on the history of a play 
and the most recent assignment by the environment to the input variables. 
A play prefix is said to be {\em consistent with a strategy} $f$ if for all states $s_j = (\inp_j, \out_j)$ in that prefix, $f(s_0,\dots, s_{j-1}, \inp_j) = \out_j$ for all $j \geq 0$. 
A strategy is memoryless if it only depends on the last state and the most recent assignment to the input variables. 
Formally, a {\em memoryless strategy} is a function   $f:(2^{\I\cup\O})\times 2^\I\rightarrow 2^\O$ such that if $(s_m,\inp')\models \rho_\I$ and $f(s_m,\inp')=\out$, then $(s_m,\inp',\out')\models \rho_\O$.

A {\em game} is a tuple $(\GS,\varphi)$ where $\GS$ is a game structure over inputs $\I$ and outputs $\O$ and $\varphi$ is
an LTL formula
over $\I\cup\O$ called a {\em winning condition}.
A play $\pi$ is \emph{winning} for the system if
$\pi \models \varphi$.
A strategy $f$ \emph{wins from state $s$} if every play $\pi$ from $s$ that is consistent with $f$ is winning for the system. A strategy $f$ \emph{wins from  $S$}, where $S$ is an assertion over $\I\cup\O$, if it wins from every state $s\models S$. The \emph{winning region} of the system is the set of states  from which it has a winning strategy. 
A strategy $f$ is \emph{winning} if for every state $\inp\models \theta_\I$ there exists a state $\out\in 2^\O$ such that $(\inp,\out)\models \theta_\O$ and $f$ wins from $(\inp,\out)$. 
In this paper, we have the following games that are defined over the following winning conditions.

\begin{itemize}
\item
{\em Reachability games}:  $\ltlF(\varphi)$ where $\varphi$ is an assertion over $\I\cup\O$. 

\item 
{\em Safety games}: $\ltlG(\varphi)$ where $\varphi$ is an assertion over $\I\cup\O$. 

\item{\em B\"uchi games}:
$\ltlG\ltlF(\varphi)$ where $\varphi$ is an assertion over $\I\cup\O$.

\item
{\em GR($k$) games}:  $\bigwedge_{l=1}^k (  
\bigwedge_{i=1}^{n_l} \ltlG\ltlF(\varphi_{l,i}) \rightarrow \bigwedge_{j=1}^{m_l} \ltlG\ltlF(\psi_{l,j}))$ where all $\varphi_{l,i}$ and $\psi_{l,j}$ are  assertions over $\I\cup\O$. 
\end{itemize}



Given a game $(\mathit{GS},\varphi)$,
\emph{realizability} is the problem of deciding whether a winning strategy for the system exists, and \emph{synthesis} is the problem of constructing a winning strategy if one exists.
We note that a realizability check can be reduced to the identification of the winning region, $W$: A winning strategy exists iff for all $\inp\models \theta_\I$ there exists $\out\in2^\O$ such that $(\inp,\out)\models \theta_\O$ and $(\inp,\out)\in W$. Hence, the synthesis problem can be solved by constructing a strategy that wins from $W$.

\subsubsection{Game Graphs and Weak B\"uchi Games.}

The \emph{game graph} for a game structure $\GS$ is the directed graph $(V, E)$ with vertices $V = 2^{\I \cup \O}$ and edges $E = \{(s, t) \mid (s, t') \models \rho_\I \land \rho_\O\}$. Intuitively, vertices are states over $\I$ and $\O$, and edges represent valid transitions between states according to the safety conditions. The game graph can be useful for analyzing the structural properties of a game structure via graph-theoretical properties.

A \emph{finite path} in a directed graph $(V, E)$ is a sequence $v_0, \ldots, v_n \in V^+$ such that $(v_j, v_{j+1}) \in E$ for all $0 \leq j < n$. An \emph{infinite path} $v_0, v_1, \ldots \in V^\omega$ is similarly defined. A vertex $u$ is said to be \emph{reachable} from another vertex $v$ if there is a finite path from $v$ to $u$.
A {\em strongly connected component} (SCC) of a directed graph $(V,E)$ is a maximal set of vertices within which every vertex is reachable from every other vertex.  It is well known that SCCs partition the set of vertices of a directed graph, and that the set of SCCs is partially ordered with respect to reachability. Also note that every infinite path ultimately stays in an SCC.

Let $(\GS, \ltlG \ltlF \varphi)$ be a game with a B\"uchi winning condition, and let $\Statess_0\dots,\Statess_m$ be the set of SCCs that partition the game graph of $\GS$. We say that $(\GS, \ltlG \ltlF \varphi)$ is a \emph{weak B\"uchi game} if, given the set $\F$ of states that satisfy the assertion $\varphi$, for every SCC $\Statess_i$, either $\Statess_i\subseteq \F$ or $\Statess_i\cap\F=\emptyset$.
Thus, the SCCs of a weak B\"uchi game are either {\em accepting components}, meaning all of its states are contained in $\F$, or  {\em non-accepting components}, meaning none of its states is present in $\F$. 
As a consequence, a play in a weak B\"uchi game is winning for the system if the play ultimately never exits an accepting component. Similarly, a strategy is winning for the system if it can guarantee that every play will ultimately remain inside an accepting component.

\section{Separated GR($k$) Games}
\label{sec:problem-def}

Our framework relies on the core idea of reducing the problem of adapter generation to synthesizing a \emph{Separated GR($k$) game}, which we define in this section. At a high-level, a separated GR($k$) differentiates from a regular GR($k$) game in a separation between input and output variables in both the game structure and winning condition. We show in later sections that the separation of variables leads to algorithmic benefits to the synthesis problem. Formally,  


\begin{definition} \label{def:separates-variables}
A game structure  $\mathit{\GS}=(\I,\O,\theta_\I,\theta_\O,\rho_\I,\rho_\O)$ {\em separates variables} over input variables $\I$ and output variables $\O$ if:
\begin{itemize}
    \item The environment's initial assumption $\theta_\I$ is an assertion over $\I$ only.
    \item The system's initial guarantees $\theta_\O$ is an assertion over $\O$ only.
    \item The environment's safety assumption $\rho_\I$ is an assertion over $\I\cup \I'$ only.
    \item The system's safety guarantee $\rho_\O$ is an assertion over $\O\cup \O'$ only.
\end{itemize} 
\end{definition}

The interpretation of a game structure which separates variables is that the underlying game graph $(V, E)$ is the product of two distinct directed graphs over disjoint sets of variables:  $G_\I$ over the variables $\I\cup\I'$, and $G_\O$ over the variables $\O\cup\O'$. For $\J \in \{\I, \O\}$, the vertices of $G_\J$ correspond to states over $\J$ and there is an edge between states $s$ and $t$ if $(s, t') \models \rho_\J$.

Next, the notion of separation of variables extends to games with GR($k$) winning conditions as follows:
\begin{definition}
\label{def:separatedGRk}
A GR(k) winning condition $\varphi$ over $\I\cup\O$  separates variables w.r.t. $\I$ and $\O$ if $ \varphi = \bigwedge_{l=1}^k (  
\bigwedge_{i=1}^{n_l} \ltlG\ltlF(\varphi_{l,i}) \rightarrow \bigwedge_{j=1}^{m_l} \ltlG\ltlF(\psi_{l,j})
)$ such that each $\varphi_{l,i}$ is an assertion over $\I$ and each $\psi_{l,j}$ is an assertion over $\O$. 
\end{definition}

A {\em Separated GR($k$) game} is a GR($k$) game $(\GS,\varphi)$ over $\I\cup\O$ in which both $\GS$ and $\varphi$ separate variables w.r.t. $\I$ and $\O$.



A major  observation is that in a game played over a separated game structure, the actions of the two players are independent: the environment's actions do no limit the system's actions, and vice versa. In later sections we see how this observation leads to algorithmic improvements in solving separated GR($k$) games over a regular GR($k$) game. 
Specifically, in Section~\ref{sec:transducers} we see how to use Separated GR($k$) games to generate the adapter transducer. In Sections~\ref{sec:reduction} and ~\ref{sec:algorithmsWeakBuchi} we discuss algorithms for realizability and synthesis of Separated GR($k$) games.




\section{From Transducers to Separated GR($k$)}\label{sec:transducers}

We describe, using an end-to-end-example, how adapter transducer generation can be reduced to synthesis of Separated GR($k$) games. 


We begin with user-provided $\Target$ and $\Adaptee$ transducers. These transducers model the behavior of a system that we want to use ($\Adaptee$) and the behavior of a system that we want to emulate ($\Target$). 
 For example, the transition systems in Figure~\ref{fig:transducers} formulates the following scenario. (1)  $\Target$ is a hardware with three modes that we use, such that the $U$ (up) and the $D$ (down) commands send the hardware from mode $s_0$ to modes $s_1$ and $s_2$, respectively, from which the $S$ (stay) command keeps the system looping at the chosen mode. (2)  $\Adaptee$ that is a hardware that we \textit{want} to use and also has three modes, but which does not allow the command $S$ after $U$. Instead, it allows a $D$ command that switches the mode back to $s_0$.



\begin{figure}
    \centering
\begin{tikzpicture}[->,>=stealth',shorten >=1pt,auto,node distance=1.5cm,initial text=,
  thick,node/.style={fill=blue!20,draw,
  font=\sffamily\footnotesize, circle},
  every node/.style={inner sep=1pt}
  ]

\node[initial,node, initial text=$S|\neg t_1 \wedge \neg t_0$] (00) {$s_0$};

\node[node] (01) [above of=00] {$s_1$};
\node[node] (10) [below of=00] {$s_2$};
\node [above=0.1cm of 01] {Target};

\draw [->] (00) -- node {$U \,|\, t_1 \wedge \neg t_0$} (01);
\draw [->] (00) -- node {$D \,|\, \neg t_1 \wedge t_0$} (10);
\draw [->] (01) edge [loop left] node {$S \,|\,\neg t_1 \wedge t_0$} (01);
\draw [->] (10) edge [loop left] node {$S \,|\, t_1 \wedge \neg t_0$} (10);

\node[initial,node, initial text=$S \,|\, \neg a_1 \wedge \neg a_0$] (_00) [left=5cm of 00] {$s_0$};
\node[node] (_01) [above of=_00]    {$s_1$};
\node[node] (_10) [below of=_00]    {$s_2$};
\node [above=0.1cm of _01] {Adaptee};

\draw [->] (_00) edge [bend right] node [right] {$U \,|\,\neg a_1 \wedge a_0$} (_01);
\draw [->] (_00) -- node {$D \,|\, a_1 \wedge \neg a_0$} (_10);
\draw [->] (_01) edge [bend right] node [left] {$D \,|\, \neg a_1 \wedge \neg a_0$} (_00);
\draw [->] (_10) edge [loop left] node {$S \,|\, a_1 \wedge \neg a_0$} (_10);

\end{tikzpicture}    
    \caption{An example of $\Adaptee$ and $\Target$  transducers. In this example, the $a_i$ and $t_i$ variables encode the binary representation of the mode being moved to. 
    }
    \label{fig:transducers}
\end{figure}
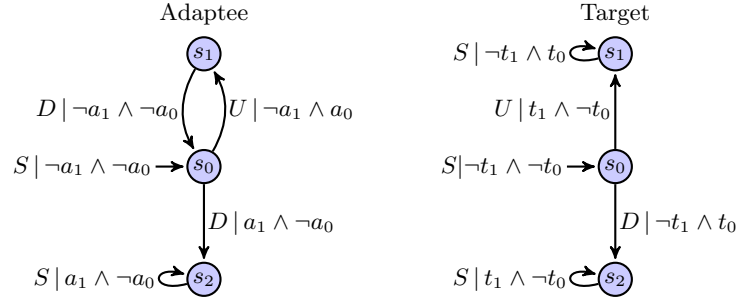

The second step is a formulation of the equivalence relation, where we define the type of emulation that we require. In our example we want to maintain the following property: if $\Target$ visits a mode $s_i$ infinitely often for a certain input sequence, then so does $\Adaptee \circ \Adapter$.  This can be expressed in LTL as: 
$$\bigwedge_{i=0}^2 \ltlG\ltlF(bin_t(s_i)) \to  \ltlG\ltlF(bin_a(s_i))$$
where $bin_t(s_i)$ denotes the binary representation of mode $s_i$ using  variables $t_1, t_0$, and similarly for $bin_a(s_i)$ using variables $a_1, a_0$.
Note that in this example we cannot just synthesize an adapter that cycles through all modes in $\Adaptee \circ \Adapter$ infinitely often, since the $\Adaptee$ transducer does not allow that.

As a third step, to generate a separated GR($k$) game,
we  translate the $\Target$  and $\Adaptee$ transducers to $\Input$ and $\Output$ transition systems as depicted, for example, in Figure~\ref{fig:transition-systems}.
Since $\Adaptee$ and $\Target$ are two separate transducers, each with its own structure, it is natural to model these as two separate transition systems on distinct variables. 
Thus, the transition systems are produced by the well known projection construction that turns an FST into a FSA that accepts the output language of the transducers~\cite{10.5555/972695.972698}. 
Note that in our setting $\Target$ is translated to $\Input$ and $\Adaptee$ is translated to $\Output$.  This may appear as a role inversion to readers.
We propose it because the role of the controller in our setting is to translate the behavior of $\Target$ to an equivalent behavior of the $\Adaptee$.

\begin{figure}
    \centering
    \begin{tikzpicture}[->,>=stealth',shorten >=1pt,auto,node distance=0.85cm,initial text=,
      thick,node/.style={fill=blue!20,draw,
      font=\sffamily\footnotesize,minimum width=10mm}]

    \node[initial,node] (00) {$\neg t_1  \wedge \neg t_0$};
    \node[node] (01) [above of=00] {$\neg t_1  \wedge t_0$};
    \node[node] (10) [below of=00] {$t_1  \wedge \neg t_0$};
    \node [above= 0.1 cm of 01] {Input};

    \draw [->] (00) --  (01);
    \draw [->] (00) --  (10);
    \draw [->] (01) edge [loop left] (01);
    \draw [->] (10) edge [loop left] (10);

    \node[initial,node] (_00) [right=2cm of 00] {$\neg a_1  \wedge \neg a_0$};
    \node[node] (_01) [above of=_00]    {$\neg a_1  \wedge a_0$};
    \node[node] (_10) [below of=_00]    {$a_1  \wedge \neg a_0$};
    \node [above= 0.1 cm of _01] {Output};

    \draw [->] (_00) edge [bend right]  (_01);
    \draw [->] (_00) --  (_10);
    \draw [->] (_01) edge [bend right]  (_00);
    \draw [->] (_10) edge [loop left] (_10);

    \end{tikzpicture}
    \caption{A direct translation of the $\Target$ transducer to an $\Input$ transition system and of the $\Adaptee$ transducer to an $\Output$ transition system.}
    \label{fig:transition-systems}
\end{figure}
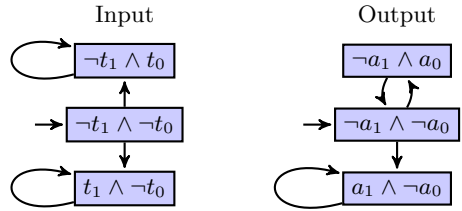

These separate transition systems, together with the specification described above, form a Separated GR($k$) that, as a fourth step, we can feed to the Separated GR($k$) synthesis algorithm.
The output of the algorithm is a transducer called $Controller$, that maps runs of $\Input$ to runs of $\Output$, as shown, in our example, in Figure~\ref{fig:controller}.  This, in fact, connects the output of the $\Target$ to the output of the $\Adaptee$.

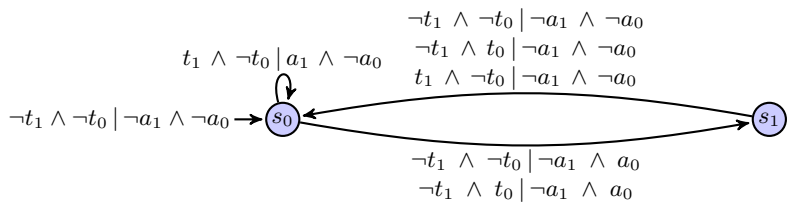
\begin{figure}
    \centering
    \begin{tikzpicture}[->,>=stealth',shorten >=1pt,auto,node distance=1cm,initial text=,
      thick,node/.style={fill=blue!20,draw,
      every text node part/.style={align=center},
      font=\sffamily\footnotesize},every node/.style={inner sep=1pt}]

    \node[initial,circle,node,initial text=$\neg t_1 \wedge \neg t_0\,|\,\neg a_1 \wedge \neg a_0$] (0) {$s_0$};
    
    \node[circle,node,minimum width=0mm] (1) [right=6cm of 0] {$s_1$};

    \draw [->] (0) edge [bend right=10] 
            node [below,text width=6cm,align=center] (A) 
            {$\neg t_1 \wedge \neg t_0 \,|\, \neg a_1 \wedge a_0$ \\ 
             $\neg t_1 \wedge t_0 \,|\, \neg a_1 \wedge a_0$}     
            (1);
    
        \draw [->] (1) edge [bend right=10] 
            node [above,text width=5cm,align=center]{$\neg t_1 \wedge \neg t_0 \,|\, \neg a_1 \wedge \neg a_0$ \\ $\neg t_1 \wedge t_0 \,|\, \neg a_1 \wedge \neg a_0$ \\ $t_1 \wedge \neg t_0 \,|\, \neg a_1 \wedge \neg a_0$} (0);
    
        \draw [->] (0) edge [loop above] 
            node [above,text width=4cm,align=center]{$t_1 \wedge \neg t_0 \,|\, a_1 \wedge \neg a_0$} (0);
    
        \end{tikzpicture}
    \caption{A controller that reads runs of the $\Input$ transition system and generates runs of the $\Output$ transition system such that the specified Separated GR($2$) formula is guaranteed to be true.}
    \label{fig:controller}
\end{figure}

As a final step, from the controller we can construct the $\Adapter$ using the formula $\Adapter = {\Adaptee}^{-1}\circ \mathit{Controller} \circ \Target$.
This means that $\Adapter$ contains an internal model of the $\Target$ and of the $\Adaptee$. These internal models are used to translate inputs to expected outputs of the adapter, then feed them to the controller, and then feed the output of the controller to the reverse of $\Adaptee$ to generate an input to $\Adaptee$ that emulates the behaviour of $\Target$.
Note that it is possible to invert transducers symbolically~\cite{10.1145/3140587.3062345}.


\subsection{Additional Usages of our Technique}\label{sec:illustrated}


We give two more examples to demonstrate  uses of Separated GR($k$). 

\subsubsection{Cleaning Robots.} This example demonstrates how one can use our technique to fulfill tasks that have not been covered by an execution of an existing transducer. Consider a cleaning robot (the $\Target$ transducer) that moves along a corridor-shaped house, from room $1$ to room $n$. The robot follows some plan and accordingly cleans some of the rooms. Our goal is to synthesize a controller that activates a second cleaning robot (the $\Adaptee$ transducer) that follows the first robot and cleans exactly those rooms left uncleaned. Each robot controls a set of variables indicating which room they are in and which rooms they have cleaned, and additionally the original robot controls a variable indicating whether it is done with its cleaning. Our controller is required to fulfill requirements of the form: $\ltlG\ltlF(\mathit{done}) \wedge \ltlG\ltlF(\mathit{!in{:}clean}_i)\rightarrow \ltlG\ltlF(\mathit{out{:}clean}_i)$, $\ltlG\ltlF(\mathit{done}) \wedge \ltlG\ltlF(\mathit{in{:}clean}_i)\rightarrow \ltlG\ltlF(\mathit{!out{:}clean}_i)$.


\subsubsection{Railway Signalling.} This example demonstrates how one can use our technique to improve the quality of an existing transducer. We consider a junction of $n$ railways, each equipped with a signal that can be turned on (light in green) or off (light in red). Some railways overlap and thus their signals cannot be turned on simultaneously. We consider an overlapping pattern where railways 1-4 overlap, and similarly 3-6, 5-8, and so on.

An existing system (the $\Target$ transducer) was programmed to be strictly safe in order to avoid accidents, so it never raises two signals simultaneously. We want to improve the system's performance by synthesizing a controller that reads the assignments that the existing transducer produces and accordingly assign values to the signals in such a way as to produce both safe and \emph{maximal} valuations: the $i$th signal is turned on if and only if the signal of every rail that overlaps with the $i$th rail is off.
Furthermore, we want to maintain liveness properties of the $\Target$ system: (1) every signal that is turned on infinitely often by the existing system must be turned on infinitely often by the new system as well, and (2) if a signal is turned on at least once every $m$ steps (where $m$ is a parameter of the specification) by the existing system, then the same holds for the new system.

Note that, in terms of the GR($k$) formula, this example is similar to the ``hardware" example that we gave; we want to emulate the $\Target$'s execution. The crux of the example lies in its $\Adaptee$. Here, unlike in the explanatory example, the $\Adaptee$ is not a given hardware, but rather a virtual component that the user introduced to improve the $\Target$ performance. In this case the $\Adaptee$ produces safe and maximal signals.

\section{Overview for Solving Separated GR($k$) Games}\label{sec:reduction}

The adapter generation framework described in Section~\ref{sec:transducers} relies on synthesizing a controller from a separated GR($k$) game. 
In this section and the next, we describe how to solve separated GR($k$) games. 
This section gives an overview of the algorithm in Section~\ref{sec:highlevel} and describes a necessary property, called the delay property, in Section~\ref{sec:delay}. The delay property is necessary to prove correctness of our synthesis algorithm. Later, Section~\ref{sec:algorithmsWeakBuchi} gives the complete algorithm and proves its correctness. 

%





\subsection{Algorithm Overview and Intuition}\label{sec:highlevel}

Following Section~\ref{sec:problem-def}, we are given a Separated GR($k$) game that consists of a game structure $\GS$ and a winning condition in a GR($k$) form $\varphi=\bigwedge_{l=1}^k \varphi_l$, where  $\varphi_l  =
\bigwedge_{i=1}^{n_l} \ltlG \ltlF(a_{l,i}) \rightarrow \bigwedge_{j=1}^{m_l} \ltlG \ltlF(g_{l,j})$. Let $G$ be the game graph of $\GS$.
Consider an infinite play $\pi$ in $\GS$. Like every infinite path on a finite graph, $\pi$ eventually stabilizes in an SCC $S$. Due to separation of variables, the game graph $G$ can be decomposed into an input graph $G_\I$ and an output graph $G_\O$. Then the projection of $S$ on the inputs is an SCC $S_\I$ in $G_\I$, and the projection of $S$ on the outputs is an SCC $S_\O$ in $G_\O$. The input side of $\pi$ converges to $S_\I$ whereas the output side $\pi$ converges to $S_\O$. 

Now, let $S$ be an SCC with projections  $S_\I$ on $G_\I$ and $S_\O$ on $G_\O$.
Then we call $S$ \textit{accepting} if for \textit{every} constraint $\varphi_l$, where $l\in\{1,\dots,k\}$, one of the following holds:
\begin{description}
\item[All guarantees hold in $S$.] For every $j\in\{1,\dots,m_l\}$, there exists $\out\in S_\O$ such that $\out\models g_{l,j}$.

\item[Some assumption cannot hold in $S$.] There exists $j\in\{1,\dots,n_l\}$ such that for all $\inp\in S_\I$, $\inp\not\models a_{l,j}$. 
\end{description}

Then from the definition of an accepting SCC we have the following: 
a strategy that makes sure that every play converges to an accepting SCC, in which all the relevant guarantee states are visited, is a winning strategy for the system in $(\GS,\varphi)$. 
To synthesize such a strategy,
we do the following: (i) synthesize a strategy $f_B$ for which every play converges to an accepting SCC; (ii) synthesize a strategy $f_\travel$ that travels within every accepting SCC, satisfying as many of the $g_{l,j}$ guarantees as possible.
(iii) construct an overall winning strategy $f$ that works as follows: the system plays $f_B$ until reaching an accepting SCC $S$, then the system switches to $f_\travel$ to satisfy as many of the $g_{l,j}$ guarantees in $S$ as possible; if the environment moves the play to a non-accepting SCC, the system can start playing $f_B$ again to reach a different accepting SCC.

The strategy $f_B$ can be found by synthesizing the weak \Buchi game  $(\GS,\ltlG \ltlF(\acc))$, where $\acc$ is the assertion that accepts exactly those states that belong to accepting SCCs (note that $(\GS,\ltlG \ltlF(\acc))$ is  a well defined  weak \Buchi game).
$f_\travel$ can be constructed by simply finding a path in $S_\O$ that satisfies the maximum number of guarantees.

A complication arises however when switching between $f_\travel$ and $f_B$, since it is conceivable that while the system is following $f_\travel$, the environment could move to a different SCC that is outside of the winning region of $f_B$.
Thus, it is not clear that we can combine these strategies to make an overall winning strategy for the system.
To show that we can indeed combine both strategies, we need the following property that we call the \textit{delay property}: if $(\inp_1, \out_1)$ is a state in the winning region of $f_B$, and $(\inp_2, \out_0)$ is a state for which there is a path in $G_\I$ from $\inp_1$ to $\inp_2$ and a path in $G_\O$ from $\out_0$ to $\out_1$, then $(\inp_2, \out_0)$ is also in the winning region of $f_B$. We formally state and prove the delay property in Section~\ref{sec:delay}. In Section~\ref{sec:algorithmsWeakBuchi} we give details of the construction of $f_B$, $f_{\travel}$ and the use of the delay property to prove correctness of the overall winning strategy $f$.

\subsection{The Delay Property}\label{sec:delay}

The delay property essentially says that if an SCC $S$ is contained in the winning region, and the environment moves from $S$ unilaterally to a different SCC $S'$, then $S'$ is also in the winning region of the system. 
In this section, we prove that the B\"uchi game $(\GS,\ltlG\ltlF(\acc))$ where $\GS=(\I,\O,\theta_\I,\theta_\O,\rho_\I,\rho_\O)$, as defined in Section~\ref{sec:highlevel}, satisfies the delay property. Throughout this section, we write $G_\I$ and $G_\O$ to denote the graphs over $2^\I$ and $2^\O$, respectively, as in Section~\ref{sec:highlevel}.
We start with the following lemma that states that the system can still win in spite of a single step delay.


\begin{lemma}\label{lem:1delay}
Let $\inp_0,\inp_1\in 2^\I$ such that  $(\inp_0,\inp'_1)\models \rho_\I$, and assume that the system can win from $(\inp_0,\out_0)$. Then the system can also win from $(\inp_1,\out_0)$.
\end{lemma}

\begin{proof}
Let $f$ be a winning strategy for the system from $(\inp_0,\out_0)$. 
We construct a winning strategy $f_d$ from $(\inp_1,\out_0)$. Intuitively, $f_d$ acts from state $(\inp_1, \out_0)$ as if it were following $f$ from state $(\inp_0, \out_0)$, with a delay of a single step: the input in the current step is used to choose the output in the next step.


We use $f$ to define $f_d$ inductively over play prefixes of length $m\geq 1$, by setting $f_d((\inp_1, \out_0), \ldots, (\inp_m, \out_{m-1}), \inp_{m+1}) = f((\inp_0, \out_0), \ldots, (\inp_{m-1}, \out_{m-1}), \inp_m)$.
Note that $f_d$ is well defined since $\GS$ separates variables: from state $(\inp, \out)$, the outputs that can be chosen for the successor state depend only on $\out$, and not on $\inp$.
Note that by this definition, for every play $(\inp_1,\out_0),(\inp_2,\out_1),\dots, (\inp_{m+1},\out_m),\dots$ consistent with $f_d$, the play $(\inp_0,\out_0),(\inp_1,\out_1),\dots, (\inp_m,\out_m),\dots$ is consistent with $f$.
We remark that we define $f_d$  only for proving the lemma, and it is \emph{not} part of our solution.

Next, we show that $f_d$ is winning from $(\inp_1,\out_0)$.
Take a play $(\inp_1,\out_0),(\inp_2,\out_1),\dots$, consistent with $f_d$. By the construction, $(\inp_0,\out_0),(\inp_1,\out_1),\dots$ is consistent with $f$. Since this is a play on a weak \Buchi game, after some point it must remain in a single SCC $S$, say from state $(\inp_j,\out_j)$. Since $f$ is a winning strategy, the SCC $S$ must be accepting. 
Then $\out_j,\out_{j+1},\dots$ is an infinite path in the SCC $S|_\O$, and $\inp_j,\inp_{j+1},\dots$ is an infinite path in the SCC $S|_\I$. 
Consequently, $(\inp_1,\out_0),(\inp_2,\out_1),\dots $ converges to an SCC $\hat{S}$ in which $\hat{S}|_\I = S|_\I$ and $\hat{S}|_\O = S|_\O$. Since the conditions for an SCC $D$ to be accepting depend only on the relation between $D|_\I$ and $D|_\O$, we have that $\hat{S}$ is accepting since $S$ is accepting as well.\qed
\end{proof}

We can now prove the delay property, following by straightforward induction from Lemma~\ref{lem:1delay}.

\begin{theorem}[Delay Property Theorem]\label{thm:delay}
Let $\inp_0,\dots,\inp_n\in (2^\I)^+$ be a path in $G_\I$, and for $m\geq 0$, let $\out_{-m},\dots,\out_0\in (2^\O)^+$ be a path in $G_\O$. Assume that the system can win from $(\inp_0,\out_0)$. Then the system can also win from $(\inp_n,\out_{-m})$.
\end{theorem}

\begin{proof}
From $(\inp_n,\out_{-m})$, the system can simply ignore the inputs and follow the path in $G_\O$ to $\out_0$. Let $(\inp_{n+m}, \out_0)$ be the state at that point in some play. Note that there is a path between $\inp_n$ and $\inp_{n+m}$, and therefore there is a path between $\inp_0$ and $\inp_{n+m}$. If the system can win from $(\inp_0, \out_0)$ then by using  Lemma~\ref{lem:1delay} in the induction steps, the system can win by induction from $(\inp, \out_0)$ for all $\inp$ such that there is a path in between $\inp_0$ and $\inp$. Therefore, the system can win from $(\inp_{n+m}, \out_0)$, and by consequence from $(\inp_n,\out_{-m})$.\qed
\end{proof}

A corollary of Theorem~\ref{thm:delay} is the following statement about the structure of the winning region of the weak \Buchi game $B=(\GS,\ltlG\ltlF (\acc))$ as defined in Section~\ref{sec:highlevel}.

\begin{corollary}\label{cor:delay}
The winning region of $B$ is a union of SCCs.
\end{corollary}

\begin{proof}
Let $(\inp,\out)$ be a state in the winning region of $B$, let $(\hat\inp,\hat\out)$ be a state in the same SCC $S$ of $(\inp,\out)$, and let $S|_\I$ and $S|_\O$ be the projections of $S$ on $G_\I$ and $G_\O$, respectively.
Then there is a path $\inp_0,\ldots, \inp_n$ for some $n\geq 0$ in $S|_\I$ such that $\inp_0=\inp$ and $\hat\inp=\inp_n$. Similarly, there is a path $\out_{-m},\ldots, \out_0$ for some $m\geq 0$ in $S|_\O$ such that $\hat\out_0=\out$ and $\hat\out=\out_{-m}$. Then by the delay property of Theorem~\ref{thm:delay}, the vertex $(\hat\inp,\hat\out)=(\inp_n,\out_{-m})$ is also in the winning region of $B$. \qed
\end{proof}

We use Theorem~\ref{thm:delay} and Corollary~\ref{cor:delay} in the proof of correctness of the overall winning strategy $f$, as described in Section~\ref{sec:symsol}.

\section{Algorithms for Solving Separated GR($k$) Games}
\label{sec:algorithmsWeakBuchi}

In this section we provide the exact details of our synthesis algorithm for Separated GR($k$) games, as described in Section~\ref{sec:highlevel}.
Since constructing $f_B$ involves defining and solving a weak \Buchi game 
, we first describe these in Section~\ref{sec:accepting-SCCs}. We remark that our weak \Buchi game synthesis algorithm works for all weak \Buchi games, and not just for the special weak \Buchi game defined in Section~\ref{sec:highlevel}. Specifically, it works even when the underlying game structure does not separates variables. Next, in Section~\ref{sec:symsol}, we complete the algorithm construction and describe the correctness of our overall synthesis algorithm. Full proofs of the theorems in this section appear in Appendix~\ref{app:StratProofs}.


\subsection{Realizability and Synthesis for Weak \Buchi Games}
\label{sec:accepting-SCCs}

We present a symbolic algorithm to solve  synthesis of a weak \Buchi game. 
When represented in explicit state-representation, weak B\"uchi games are known to be solved in linear-time in the size of the game~\cite{Chatterjee08,LodingT00}. In this section, we adapt the algorithm from~\cite{Chatterjee08,LodingT00} to symbolic state-space representation. 
For sake of exposition, we give an overview of the algorithm and then present our symbolic modification.

\subsubsection{Overview}

Given a weak B\"uchi game, recall that each SCC in its game graph $G$ is either an accepting SCC or a non-accepting SCC.
The goal is to find the winning regions in the weak B\"uchi game. This can be done by backward induction on the topological ordering of the SCCs as follows. 
Let $(\Statess_0, \dots \Statess_m)$ be a topological sort of the SCCs in $G$. 

\textbf{Base Case:}
Consider all {\em terminal partitions}, say $\Statess_j,\dots,\Statess_m$; that is, every SCC  from which no other SCC is reachable.
In this case, plays beginning in a terminal SCC will never leave it. Therefore, all states of terminal SCCs that are accepting are in the winning region of the system and all states of terminal SCCs that are non-accepting are not in the winning region of the environment. 


\textbf{Induction Step:} Let $\vec{S}=(S_{i+1},\dots,S_m)$, and
suppose that the set $\bigcup \vec{S}$ has been classified into winning regions  for the system $W_{i+1}^s$ and the environment $W_{i+1}^e$, respectively.
 Let $\vec{S}_\mathit{new}=(\Statess_j,\Statess_{j+1}, \ldots, \Statess_{i})$ be the SCCs from which all edges leaving the SCC lead to an SCC in $\vec{S}$. 
 Further, let $A$ and $N$ be the unions of all accepting SCCs and all non-accepting SCCs in $\vec{S}_{new}$, respectively. 
 Then the basic idea is as follows: The system can win from $s\in N$ if and only if it can force $\ltlF(W_{i+1}^s)$ from $s$. Analogously, the system can win from $s\in A$ if and only if it can force  $\ltlG(A\cup W_{i+1}^s)$ from $s$. Hence, by solving these reachability and safety games, we can update $W_{i+1}^s$ and $W_{i+1}^e$ into $W_{j}^s$ and $W_j^s$ that partition the larger set  $\bigcup(S_j,\dots,S_m)$ into winning regions for the system and the environment. 
 The winning strategy can be constructed in a standard way as a side-product of the reachability and safety games in each step, see for example~\cite{ZhuTLPV17,ZhuTLPV17b}.

\subsubsection{Symbolic Algorithm for Weak \Buchi Games.} 

Given a weak \Buchi game $B = ((\I, \O, \theta_\I, \theta_\O, \rho_\I, \rho_\O), \ltlG \ltlF (\acc))$
with BDDs representing $\theta_\I$, $\theta_\O$, $\rho_\I$, $\rho_\O$ and $\acc$, our goal
is to compute a BDD for the winning region and to synthesize a memoryless winning strategy for the system. 
The construction follows a fixed-point computation that adapts the inductive procedure described in the overview: In the basis of the fixed point computation, the winning region is the set of accepting terminal SCCs; in the inductive step, the winning region includes winning states by examining SCCs that are higher in the topological ordering on SCCs.
In what follows we describe a sequence of BDDs that we construct towards constructing the overall BDD for the winning region.
We use the notation $\X$ 
to denote a set of variables over $\I\cup\O$. 
For the sake of the current construction, memoryless strategies are given in the form of BDDs over $\X,\X'$, for further details on the BDDs constructions refer to Appendix~\ref{app:prelims}.


\paragraph{BDD constructions.}\label{par:BDD-constructions} We start by constructing a BDD for a predicate that indicates whether two states in a game structure are present in the same SCC. Let predicate $\Reach(s,t')$ hold if there is a path from state $s$ over $\I\cup\O$ to state $t$ over $\I\cup\O$ in the game structure $\GS$.
Similarly, a predicate $\Reachinverse(s,t')$ holds if and only if  $\Reach(t,s')$ holds.
BDDs for   $\Reach$  and $\Reachinverse$  can be computed in $O(N)$ symbolic operations using the transition relation of the game structure.
Then, a BDD indicating if two states share the same SCC, is constructed in  $O(N)$ symbolic operations by    $\SCC(\X,\X') := \Reach(\X,\X') \wedge \Reachinverse(\X,\X')$.

Next, we construct a BDD for the union of the terminal SCCs, required by the basis of induction for the construction of the winning region. Let predicate $\terminal (s)$ hold if state $s$ over $\I\cup\O$ is present in a terminal SCC. Then
$\terminal(\X) := \forall \X': \Reach(\X,\X')\rightarrow \SCC(\X,\X')$.
Therefore, given BDDs for $\mathsf{Reach}$ and $\mathsf{SCC}$, the construction of $\mathsf{Terminal}$ requires $O(1)$ symbolic operations.

\paragraph{Computing the Winning Region.}

We now describe the fixed-point computation to construct a BDD for the winning region in  a weak \Buchi game. 
Let $\Reachability_{(M,N)}(\X)$ denote a BDD generated by solving a reachability game that takes as input a set of source states $M$ and target states $N$ and outputs those states in $M$ from which the system can guarantee to move into $N$.
Similarly, let $\Safety_{(M,N)}(\X)$ denote a BDD generated by solving a safety game that takes as input a set of source states $M$ and target states $N$ and outputs those states in $M$ from which the system can guarantee that all plays remain inside the set $N$.
These constructions are standard, details can be found in~\cite[Chapter~2]{2001automata}.

Now, let  $\winning(s)$ denote that state $s$ over $\I\cup\O$ is in the winning region. Then, $\winning(\X)$ is the fixed point of the BDD $\winningAux$ defined below, where the construction essentially follows the high-level algorithm description. The BDD $\Acc(\X)$ represents the formula $acc$ encoding the set of accepting states. In addition, $\C^i(\X)$  is the union $\bigcup \vec{S}$ of the Downward-Closed set of SCCs, i.e. the SCCs that have already been classified into winning or not-winning, and $\C^i_{new}(\X)$
is the union $\bigcup \vec{S}_{new}$ of the SCCs in $\C^i(\X)$ that were not in $\C^{i-1}(\X)$. Finally, $\NC^i(\X)$ is the subset $N$ of non-accepting states in $\C^i_{new}(\X)$, and $\AC^i(\X)$ is the subset $A$ of accepting states in $\C^i_{new}(\X)$. We then define $\winningAux$ as follows.
\begin{description}
    \item{\textbf{Base Case.}} $\winningAux^0(\X) := \terminal(\X) \wedge \Acc(X)$ and
    $\C^0(\X) := \terminal(\X)$
    
    \item {\textbf{Inductive Step.}}
    \begin{align*}
    \C^{i+1}(\X) := & \forall \X': \Reach(\X,\X') \rightarrow (\SCC(\X,\X') \vee \C^i(X'))\\
    \C^{i+1}_{new}(\X) := & \C^{i+1}(\X) \setminus \C^i(\X) \\
    \NC^{i+1}(X):= & \C^{i+1}_{new}(X) \wedge \neg\Acc(X)\\
    \AC^{i+1}(X):= & \C^{i+1}_{new}(X) \wedge \Acc(X)
    \end{align*}
    \begin{align*}
    \winningAux^{i+1}(\X) & :=        \winningAux^i(\X) \\
    & \vee \Reachability_{(\NC^{i+1}(X), \winningAux^i(\X))}(X) \\
    &  \vee  
    \Safety_{(\AC^{i+1}(X), \AC^{i+1}(X)\vee\winningAux^i(\X))}(X)
    \end{align*}
  
\end{description}

To explain the construction of $\winning$, note that a state $s$ in $\C^{i+1}(\X)$ is winning in one of these cases:
(i) $s$ is a winning state in $\C^i(\X)$.  
(ii) $s$ is a non-accepting state in $\C^{i+1}(\X)$ from which the system can force the play into a winning state in $\C^i(\X)$. This set of states can be obtained from $\Reachability_{(\NC^{i+1}(X), \winningAux^i(\X))}(X)$. 
(iii) $s$ is an accepting state in $\C^{i+1}(\X)$ from which the system can guarantee that every play that leaves the accepting SCC moves into a winning state in $\C^i(\X)$. This set of states can be obtained from $\Safety_{(\AC^{i+1}(X), \AC^{i+1}(X)\vee\winningAux^i(\X))}(X)$. 

Finally, to check realizability,  construct the BDD $\forall \I ( \InitIn(\I) \rightarrow \exists \O ( \InitOut(\O) \land \winning(\I \cup \O)))$, where $\InitIn(\I)$ and  $\InitOut(\O)$ are BDDs representing $\theta_\I$ and $\theta_\O$, respectively. This BDD is equal to $true$ iff $B$ is realizable.

The fixed-point computation can be extended in a standard way to also compute a BDD representation $\strategy(X, X')$ of the winning strategy $f_B$, such that $(s, (i', o')) \models \strategy(X, X')$ iff $f_B(s, i) = o$. See Appendix~\ref{app:symbolicweakbuchi} for details. We then have the following theorem that follows our construction.

\begin{theorem}\label{thm:weakbuchi}
Realizability and synthesis for weak \Buchi games can be done in $O(N)$ symbolic steps.
\end{theorem}

\paragraph{Proof Outline.}
The proposed construction symbolically implements the inductive procedure of the explicit algorithm. Hence, it correctly identifies the system's winning region. 
It remains to show that the algorithm performs $O(N)$ symbolic operations. First of all, the constructions of $\SCC$ and $\terminal$ take $O(N)$ symbolic operations collectively. It suffices to show that in the $i$-th induction step, solving the reachability and safety games  performs  $O(|\C^{i+1}\setminus\C^i|)$ operations. This can be proven by a careful analysis of the operations and the sizes of resulting BDDs using standard results on safety and reachability games. Details have been deferred to Appendix~\ref{app:symbolicweakbuchi}.\qed

\subsection{Realizability and Synthesis for Separated GR($k$) Games}\label{sec:symsol}

We finally make use of the elements obtained so far towards solving synthesis for Separated GR($k$) games. 
Our construction follows the overview from Section~\ref{sec:highlevel}. 
To recall, we describe and construct two auxiliary strategies $f_B$ and $f_{\travel}$ and combine them to generate the final strategy $f$. 
We use the delay property theorem from Section~\ref{sec:delay} to prove the correctness of our algorithm. 


We are given a Separated GR($k$) game structure $\GS = (\I,\O,\theta_\I,\theta_\O,\rho_\I,\rho_\O)$ and a winning condition $\varphi=\bigwedge_{l=1}^k \varphi_l$, where  $\varphi_l  =
\bigwedge_{i=1}^{n_l} \ltlG \ltlF(a_{l,i}) \rightarrow \bigwedge_{j=1}^{m_l} \ltlG \ltlF(g_{l,j}))$. We first represent $\GS$ and $\varphi$ as BDDs by standard means.
We then define and construct the following.

\paragraph{Constructing $f_B$.}
Auxiliary strategy $f_B$ is the winning strategy of the system player in a weak B\"uchi game constructed form the separated GR($k$) game.
To construct a weak \Buchi game, we first construct, in $O(|\varphi|+N)$ symbolic operations, a BDD $\Acc(\I \cup \O)$ that describes the set of accepting states. The construction is standard, see Lemma~\ref{lem:acc} in Appendix~\ref{appx:grk-strategy} for details. Next, let $acc$ be the assertion represented by $\Acc$ (the assertion defined in Section~\ref{sec:highlevel}). Then the Weak \Buchi game is $B=(\GS, \ltlG \ltlF (\acc))$. Finally, we construct $f_{B}$ as the winning strategy of $B$, following Section~\ref{sec:accepting-SCCs} and Section~\ref{app:symbolicweakbuchi}.

\paragraph{Constructing $f_{\travel}$.}
For the construction of $f_{\travel}$, we arbitrarily order all guarantees that appear in our GR($k$) formula: ${\gar}_0,\dots,{\gar}_{m-1}$. For each guarantee $\gar_j$, we construct a reachability strategy $f_{r(j)}$ that, when applied inside an SCC $S_\O$ in the output game graph $G_\O$, moves towards a state that satisfies $\gar_j$ without ever leaving $S_\O$. In case no such state exists in $S_\O$, $f_{r(j)}$ returns a distinguished value $\bot$. Note that this strategy can entirely ignore the inputs.
We equip $f_\travel$ with a memory variable $\mathit{mem}$ that stores values from $\{0,\dots,m-1\}$. Then $f_\travel(s,i)$ is operated as follows: for $\mathit{mem},\mathit{mem}+1,\dots$ we find the first $\mathit{mem}+j \pmod{m}$ such that the SCC of $s$ includes a $\gar_j$-state, and   activate $f_{r(\mathit{mem}+j)}$ to reach such state. If no guarantees can be satisfied in $S$, we just return an arbitrary output to stay in $S_\O$. 
The construction of $f_\travel$ requires $O(|\varphi|N)$ symbolic BDD-operations as we need to construct $m$ reachability strategies (clearly, $m\leq |\varphi|$).

\paragraph{Constructing the overall strategy $f$.}

Finally, we interleave the strategies $f_B$ and $f_\travel$ into a single strategy $f$ as follows:
given a state $s$ and an input $i$, if $s \models \Acc(X)$ (that is, if $s$ is an accepting state), then set $f(s,i)=f_\travel(s,i)$; otherwise set $f(s,i) = f_B(s,i)$. Whenever $f$ switches from $f_B$ to $f_\travel$, the memory variable $\mathit{mem}$ is reset to $0$.
The next lemma proves that if $f_B$ is winning then so is $f$. 

\begin{lemma}\label{lem:win} 
If $f_B$ is a winning strategy for the weak \Buchi game $B = (\GS,\ltlG \ltlF(\acc))$, then $f$ is a winning strategy for the Separated GR($k$) game $(\GS, \varphi)$.
\end{lemma}

\begin{proof}
Since $f_B$ is a winning strategy, then for every initial input $\inp \models \theta_\I$ there is an initial output $\out \models \theta_\O$ such that $(i,o)$ is in the winning region of $GS$.
We show that playing $f$ always keeps the play in the winning region of $GS$, and therefore the play eventually converges to an accepting SCC. Once this happens, following $f_\travel$ guarantees that $\varphi$ is satisfied.
We know that as long as the play is in the winning region of $B$, following $f_B$ will keep it inside the winning region. Therefore, when we switch from $f_B$ to $f_\travel$ we must be inside the winning region and, by definition of $f$, in some accepting SCC $S$.
Then $f_\travel$ makes sure that as long as the environment remains in $S|_\I$, the projection of $S$ over the inputs, the system remains in $S|_\O$, the projection of $S$ over the output. Thus all in all the play remains in the winning region of $S$.

Therefore, the only way that the play can leave the winning region is if, when the system is in a state $(\inp_0, \out_0)$ and chooses some output $\out_{-m}$ according to $f_\travel$, the environment chooses input $\inp_n$ such that the play leaves $S$ and moves to a state $(\inp_n, \out_{-m})$ in a different SCC of $G$.
Note, however, that in this case there is a path from $\inp_0$ to $\inp_n$ and a path from $\out_{-m}$ to $\out_0$ (since by construction $f_\travel$ remains in the same SCC in $G_\O$). Since $(\inp_0, \out_0)$ is in the winning region, by Theorem~\ref{thm:delay} we have that  $(\inp_n, \out_{-m})$ is in the winning region as well. 
\qed
\end{proof}






\subsubsection{Final Results.}

Given Lemma~\ref{lem:win}, we can obtain our final results on synthesis and realizability of Separated GR($k$) games, as follows.
Given a Separated GR($k$) game $(\GS,\varphi)$, construct $\acc$ and solve the weak \Buchi game $(\GS,\ltlG \ltlF(\acc))$. Then construct $f_B$, $f_{travel}$ and $f$ as described above. If realizable, then $f_B$ is a winning strategy and
from Lemma~\ref{lem:win} we have that $f$ is a winning strategy for $(\GS,\varphi)$.
If $(\GS,\ltlG \ltlF(\acc))$ is unrealizable, then the environment can force every play to converge to a non-accepting SCC. Since the GR($k$) winning condition cannot be satisfied from a non-accepting SCC, $(\GS,\varphi)$ is also not realizable. Thus we have the following theorem, see Appendix~\ref{appx:grk-strategy} for details.

\begin{theorem}\label{thm:realiizability}
Realizability for separated GR($k$) games can be reduced to realizability of weak \Buchi games.
\end{theorem}

The final result on solving Separated GR($k$) games is as follows:

\begin{theorem}\label{thm:main}
Let $(\mathit{\GS},\varphi)$ be a separated GR($k$) game over the input/output variables $\I$ and $\O$, respectively. Then, the realizability and synthesis problems for $(\mathit{\GS},\varphi)$ are solved in $O(|\varphi|+N)$ and $O(|\varphi|N)$ symbolic operations, respectively, where $N=|2^{\I\cup\O}|$.
\end{theorem}

\begin{proof}
Following Lemma~\ref{lem:win} and Theorem~\ref{thm:realiizability}, it is left to analyze the number of symbolic operations for constructing $f_B$ and then $f$.
In symbolic operations, constructing $\acc$ takes $O(|\varphi|+N)$, and computing the winning region $W$ for $(\GS,\ltlG \ltlF(\acc))$ takes $O(N)$. Checking realizability can be done by checking if for every initial input $i$ there is an initial output $o$ such that $(i, o) \in W$, which takes $O(1)$. The winning strategy $f_B$ can be computed in the process of computing $W$, taking the same number of operations (see Appendix~\ref{app:symbolicweakbuchi} for details). Finally, constructing $f_{\travel}$ takes $O
((\#\mathit{gars})N)\leq O(|\varphi|N)$, where $\mathit{gars}$ are all guarantees $\ltlG\ltlF(g_{i,\ell})$ that appear in $\varphi$. Therefore, constructing $f$ takes $O(|\varphi| N)$ symbolic operations in total.
\qed
\end{proof}

This is an improvement over the complexity of GR($k$) games, in general~\cite{PitermanP06}.
\section{Implementation and Evaluation}
\label{sec:experiments}

We have implemented our Separated GR($k$) framework for realizability and synthesis in a prototype tool \tool{}. The tool implements our symbolic algorithm using the $\mathsf{CUDD}$~\cite{CUDD} package for BDD manipulation. Our tool is evaluated on a suite of benchmarks created from the examples described in Section~\ref{sec:transducers}.




\subsubsection{Benchmark Suite.}

We have created a suite of parametric benchmarks from the three examples described in Section~\ref{sec:transducers}. Our suite consists of 38 realizable specifications. The parametric versions of the examples are described here. 

The \emph{multi-mode hardware} example is a generalization of the  example presented at the beginning of Section~\ref{sec:transducers}. It is parameterized by the number of bits $n$ and has $2^n$ modes. The \emph{Target} can move from mode $0$ to any mode and stay there, while the \emph{Adaptee} can only move from mode $0$ to odd-numbered modes, and up and down between modes $2i$ and $2i+1$. The specification consists of $2n$ variables. We generate 10 such benchmarks with $n\in\{1,\ldots,10\}$.

The {\em cleaning robots} example is parameterized in the number of rooms.
For a scenario with $n$ rooms, the specification is written over $4n+1$ variables. 
We create 10 such benchmarks with $n\in\{1\ldots,10\}$.


The {\em railways signalling} example consists of two parameters: a junction of $n$ railways and the frequency parameter $m$.
With parameters $n$ and $m$, the specification consists of  $(2+2\lceil \log m\rceil)n$ variables. We generate 18 benchmarks with $n \in \{2,\ldots, 10\}$ and $m \in \{2,3\}$.

\subsubsection{Experimental Setup and Methodology.}

We evaluate our tool against \textsf{Strix}~\cite{StrixWeb,MeyerSL18,LuttenbergerMS20}, the current state-of-the-art tool for $\ltl$ synthesis and SYNTCOMP 2020 winner of 3 out of 4 tracks~\cite{SyntComp}. In order to run our benchmarks on \textsf{Strix}, we transform the benchmarks (a game structure and a winning condition) into an $\ltl$ formula that characterizes the same winning plays using the strict semantics from~\cite{Jacobs016}. 
To the best of our knowledge, there is no other synthesis/realizability tool that operates on GR($k$) specifications. 

We compare the running time
for checking realizability. For this, we compare the runtime of realizability checks of each benchmark on both tools. Every benchmark is tested 10 times on both tools. We do this to account for the randomness introduced during BDD construction due to the automatic variable ordering by \textsf{CUDD}.
For each benchmark we evaluate (a). the number of executions on which the tools terminate and (b). the mean running time over 10 executions. 
 
All experiments were executed on a single node of a high-performance computer cluster consisting of an Intel Xeon processor running at 2.6 GHz with 32GB of memory with a timeout of 10 mins.




\begin{figure}[t]
    \centering
    \begin{subfigure}{0.5\textwidth}
        \centering
        \includegraphics[width=\textwidth]{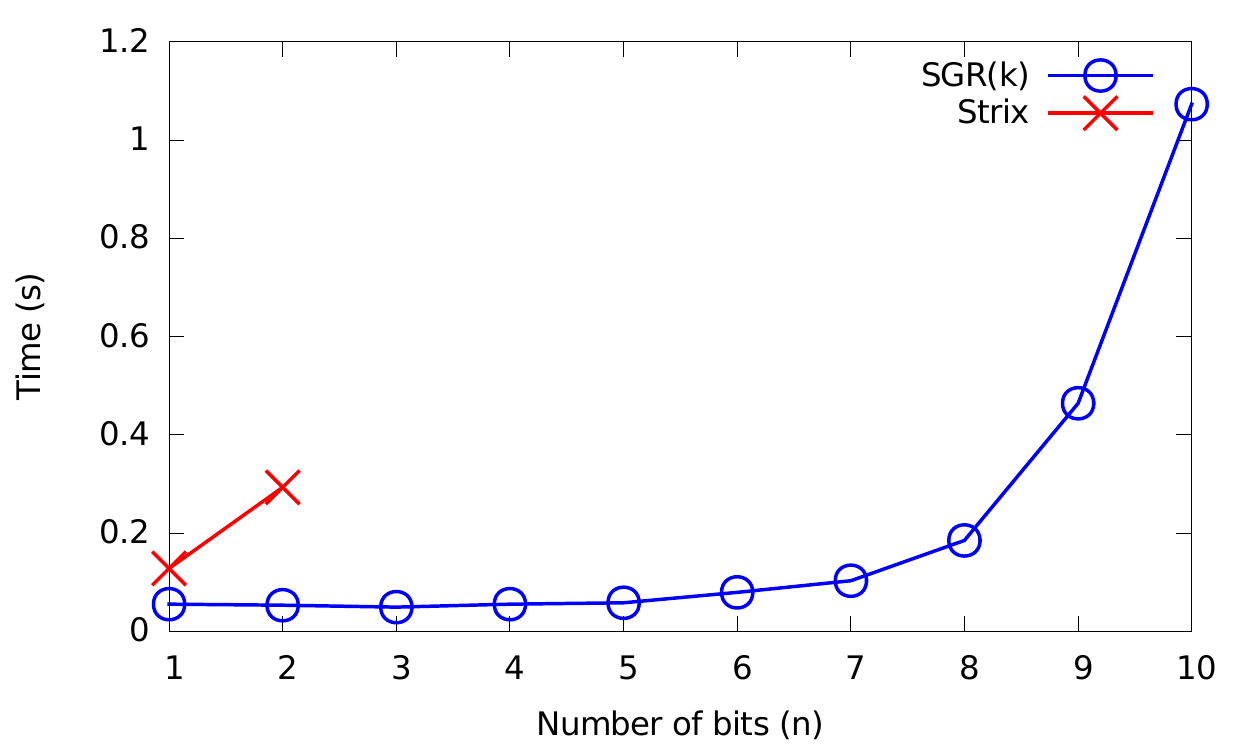}
        \caption{Multi-Mode Hardware}
         \label{fig:cleaning-robots}
    \end{subfigure}%
    \begin{subfigure}{0.5\textwidth}
        \centering
         \includegraphics[width=\textwidth]{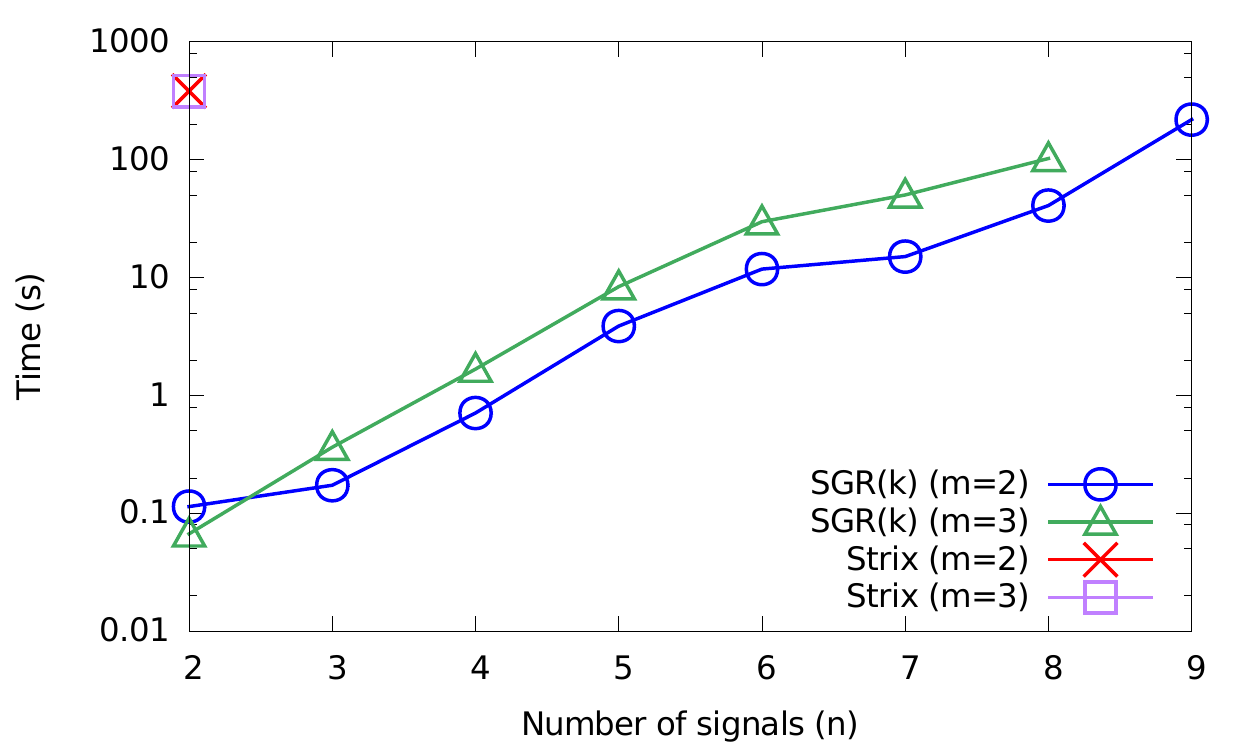}
        \caption{Railway Signalling}
        \label{fig:railway-signaling}
    \end{subfigure}

   \begin{subfigure}{0.5\textwidth}
        \centering
        \includegraphics[width=\textwidth]{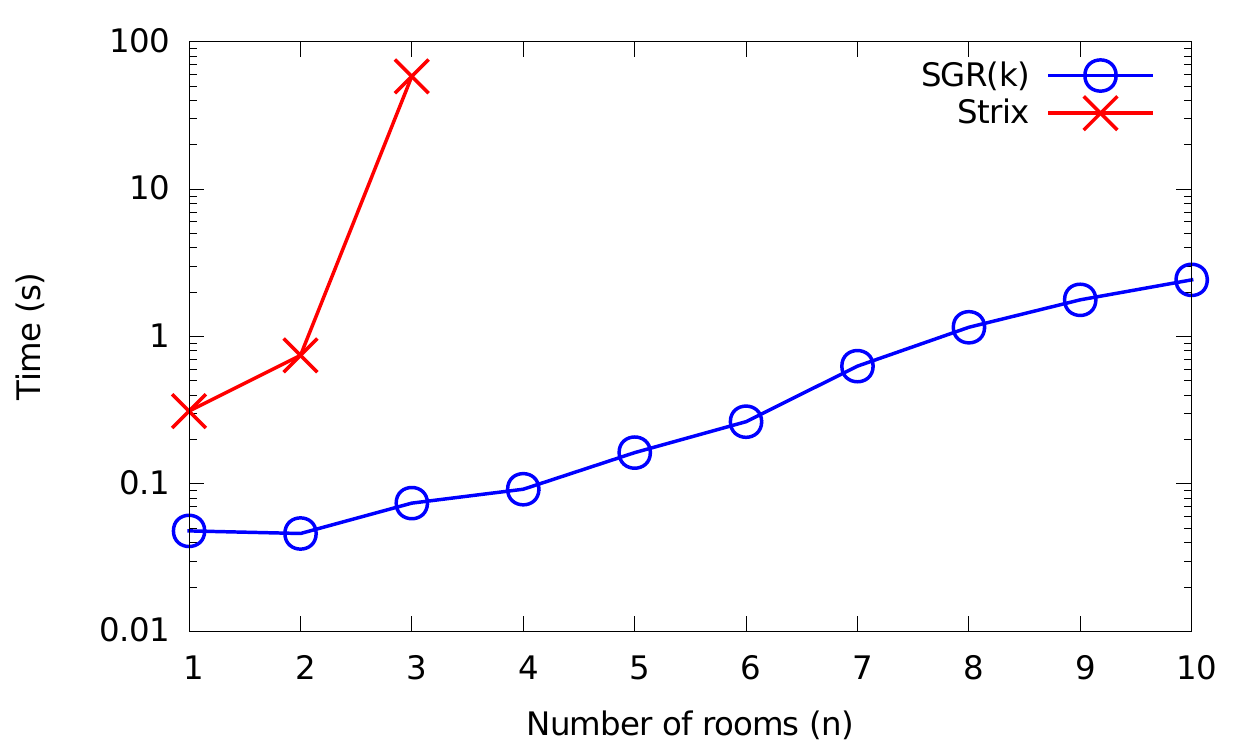}
       \caption{Mean Cleaning Robots running time \label{fig:robots}} 
\end{subfigure}
    
        \caption{Mean running time for different classes of benchmarks.}
        \label{fig:running-time}
\end{figure}

\subsubsection{Observations and Inferences.}


Our experiments clearly demonstrate the scalability and efficiency of our tool in solving Separated GR($k$) formulas.

\begin{table}[b!]
    \centering
    \begin{tabular}{|c|c||c|c|c|c|c|c|c|c|c|c|}
    \hline
         \multicolumn{2}{|c||}{$n$} & 1 & 2 & 3 & 4 & 5 & 6 & 7 & 8 & 9 & 10   \\ \hline \hline

        \multirow{2}{*}{$\sf{MultiMode(\text{$n$})}$} & \tool{} & 0.06 & 0.05 & 0.05 & 0.06 & 0.06 & 0.08 & 0.1 & 0.19 & 0.46 & 1.07 \\
         \cline{2-12}
         & Strix & 0.13 & 0.29 & TO & TO & TO & TO & TO & TO & TO & TO \\ \hline \hline

         \multirow{2}{*}{$\sf{Cleaning(\text{$n$})}$} & \tool{} & 0.05 & 0.05 & 0.07 & 0.09 & 0.16 & 0.26 & 0.63 & 1.16 & 1.78 & 2.43 \\
         \cline{2-12}
         & Strix & 0.31 & 0.75 & 58.3 & TO & TO & TO & TO & TO & TO & TO \\
         \hline
         
         \hline \hline
         
        \multirow{2}{*}{$\sf{Railways(\text{$n,2$})}$} & \tool{} & - & 0.11 & 0.17 & 0.71 & 3.88 & 11.8 & 15.1 & 40.8 & 219 & TO  \\
         \cline{2-12}
         & Strix & - & 382 & TO & TO & TO & TO & TO & TO & TO & TO \\
         \hline
         
        \multirow{2}{*}{$\sf{Railways(\text{$n,3$})}$} & \tool{} & - & 0.07 & 0.36 & 1.67 & 8.39 & 29.8 & 50.3 & 102 & TO & TO \\
         \cline{2-12}
         & Strix & - & 381 & TO & TO & TO & TO & TO & TO & TO & TO \\
         \hline

    \end{tabular}
    
    \caption{Mean realizability check running times (sec.)}
    \label{tbl:results}
\end{table}

\autoref{fig:running-time} plots the mean running time for the three  benchmarks. We further report the mean values in Table~\ref{tbl:results}. The table rows refer to the benchmarks we examine, and the columns refer to the value of the parameter $n$. As an example, for the specification $\sf{Cleaning(\text{3})}$, \tool{}'s mean running time is 0.07 sec. (row titled $\sf{Cleaning(\text{n})}$;\tool{}, column titled 3) and $\mathsf{Strix}$'s median realizability check running time is 58.3 sec. (row titled $\sf{Cleaning(\text{n})}$;Strix), column titled 4). Cells reading `TO' indicate experiments reached a timeout.


The results show that our tool solves a significantly larger number of benchmarks than \textsf{Strix}. On the few benchmarks which \textsf{Strix} solves, our tool outperforms it by several orders of magnitude. Although the running time may vary depending on the automatic variable ordering chosen by \textsf{CUDD}, we do not believe it would vary enough to significantly change the results. Specifically, we calculated the 99\% confidence interval for our results, and validated that for all data points our tool's entire interval lies below the entire interval for \textsf{Strix}.

Only three benchmarks were unsolvable by our tool (in the sense that the majority of the 10 executions timed out).
The three benchmarks are the railway signal examples with $(n=10,m=2)$, $(n=9,m=3)$, and $(n=10,m=3)$. These benchmarks consist of a large number of variables (54, 40, and 60, respectively), making them particularly challenging. All executions of the remaining benchmarks were solved in less than 4 mins by our tool.




We also examined the number of solved executions per benchmark. Our tool solved all 10 executions for 35 out of 38 benchmarks. These are the 35 benchmarks that appear as solved in \autoref{fig:running-time}. For the railway signalling benchmark with $(n=10,m=2)$, our tool solved 2 out of 10 executions. In contrast, $\textsf{Strix}$ was not able to solve even one execution for 31 out of 38 benchmarks. Even increasing the timeout to 8hrs only allowed \textsf{Strix} to solve a single additional benchmark. In total, $\mathsf{Strix}$ and our tool verified realizability of 7 benchmarks and 36 out of 38 benchmarks, respectively.

In summary, our experiments demonstrate that our tool is able to solve specifications which are challenging for existing tools.
\section{Related Work}\label{sec:related}


The Adapter design pattern was introduced in~\cite{gamma1995design}, and has been used in many
software contexts since. Our interpretation of the pattern is inspired by automata
based description of the pattern proposed by Pedrazzini~\cite{10.1007/3-540-48057-9_19}. We reformulated the problem as synthesis of reactive controllers that compose with existing systems to achieve a temporal specification. This is an active area of research~\cite{bansal2019synthesis,d2013synthesizing,ciolek2016interaction,berardi2003automatic,alur2016compositional,FLOV18}. Our work differs from existing frameworks in its variables separation feature. \emph{Shield synthesis} is a notable similar problem in which a synthesized controller corrects safety violations of an existing
controller~\cite{KonighoferABHKT17}. In contrast, our problem is mostly concerned about liveness adaptation.

Reactive synthesis of $\ltl$ winning conditions is 2EXPTIME complete in the size of the formula~\cite{pnueli1989synthesis}, making it difficult to scale for applications. 
An approach to overcome the computational barrier has been to investigate fragments and variants of $\ltl$ with lower complexity for synthesis~\cite{AmramMP19,Ehlers11,Zielonka98,de2013linear,lustig2009synthesis}.
The GR($k$) is one such fragment~\cite{BloemJPPS12}. It offers a balance between efficiency and expressiveness.
GR($k$) games are known to be efficient as they are solved in exponential time in the number of conjunctions $k$ rather than exponential in the state-space~\cite{PitermanP06}. 
Several studies have also shown they are highly expressive. As evidence, all properties expressed by deterministic B\"uchi automata (DBA) can be expressed in GR($k$)~\cite{Ehlers11}. A study of commonly appearing $\ltl$ patterns has shown that 52 of 55 patterns are DBA properties~\cite{DwyerAC99,MaozR16}. DBA properties have also been identified as common patterns in robotics applications~\cite{MenghiTPGB19}.

Finally, Separated GR($k$) games exhibit the \textit{delay property}. Intuitively, this property means that the system can win even after delaying its action for a finite amount of time while ignoring the environment before ``catching up" with the environment. While this is reminiscent of asynchrony in reactive systems~\cite{BNS18,PnueliR89,ScheweF06}, a further exploration of relations between asynchrony and the  delay property is required.

\section{Conclusion}\label{sec:Discussion} 

 
This paper presents a reactive systems-based model of the adapter design pattern. We model the adapters as transducers and reduce the problem of finding an $\Adapter$ transducer for a given $\Adaptee$ and $\Target$ systems, to the problem of synthesizing strategies for Separated GR($k$) games. Through an analysis of theoretical complexity and algorithmic performance, we show that realizability and synthesis of Separated GR($k$) games is efficient and scalable. Furthermore, by outperforming Strix, an existing state-of-the-art synthesis tool, we show that algorithms for the Separated GR($k$) class of specifications add value to the portfolio of reactive synthesis tools.


The benefits of separation of input and output variables were previously shown in the context of Boolean Functional Synthesis~\cite{CFTV18}. Through this work, we showed that the separation leads to practically viable solutions even in the context of temporal reactive synthesis, more specifically when encoding the types of equivalence relations that appear in reactive adaptation, where properties
of runs of the first system are compared to properties of runs of the other. Since the
systems may be loosely coupled, i.e., they may not run on the same clock, specifications
that impose joint temporal constraints on the two systems may not be realizable. Thus, our proposition to use the type of equivalence that separated GR($k$) formulas allow, gives users the power needed for
comparing the \textit{overall behaviors} of the systems while allowing realizability and efficient
synthesis.

The results presented in this paper encourage future studies on the separation of variables in a broader context. For instance, we could reason about variants of the adapter design pattern that do not separate variables all the way through. That is to say,  variants that translate to more general GR($k$) specifications in which the separation appears in the input and output systems but not in the specification itself.  One could further study the notion of separation of variables in more the general LTL specifications. Another direction is to consider systems that gets two types of input: from the existing system (i.e. the $\Target$)
as well as from an environment.
We believe that these future directions would enable the development of tools for synthesis from temporal specifications with a focus on expressing realistic applications as well as ensuring scalability and efficiency for practical needs. 






\subsubsection*{Acknowledgements}
We thank Supratik Chakraborty and Dana Fisman for their inputs at various stages of the project. 
This work is  supported in part by NSF grant 2030859 to
the CRA for the CIFellows Project, NSF grants IIS-1527668, CCF-1704883,
IIS-1830549, and an award from the Maryland Procurement Office. It was also supported by ISF grant 2714/19 and and
by the Lynn and William Frankel Center for Computer Science.

%
%
 \bibliographystyle{splncs04}
 \bibliography{SeparatedSpecification}

\clearpage

\appendix

\section{Binary Decision Diagrams}\label{app:prelims}



A \emph{Binary Decision Diagram} (BDD)~\cite{Akers78,Bryant86} is a DAG representation of an assertion. The non-leaves BDD nodes are labeled with variable names, the BDD arrows are labeled variable valuations, and BDD leaves with $\mathit{true}$/$\mathit{false}$. BDDs are considered a useful representation as they support Boolean operations and further satisfaction queries that can be computed efficiently~\cite{Bryant92}. Various professional tools that implement BDDs exist e.g.~\cite{CUDD,Budddy,ABCD,CAL,Milvang-JensenH98}.

We say that a BDD $\mathsf{B}$ is over $\mathcal V_1,\dots,\mathcal V_n$, if it represents an assertion over $\mathcal V_1{\cup}\cdots{\cup }\mathcal V_n$, and write $\mathsf{B}(\mathcal V_1,\dots,\mathcal V_n)$. For a BDD $\mathsf{B}(\mathcal V_1,\dots,\mathcal V_n)$ and states $s_1,\dots,s_n$ over $\mathcal V_1,\dots,\mathcal{V}_n$, respectively, we write $(s_1,\dots,s_n)\models \mathsf{B}(\mathcal V_1,\dots,\mathcal V_n)$ if the assertion that
$\mathsf{B}(\mathcal V_1,\dots,\mathcal V_n)$ represents is evaluated to $\mathit{true}$ by the states $s_1,\dots,s_n$. For a set of variables $\mathcal V$, we write $\mathcal V=\mathcal V'$ to denote the BDD that represents the assertion $\bigwedge_{v\in \mathcal V}(v=v')$. We also employ the following \emph{symbolic operations}:
\begin{enumerate}
    \item Boolean operations over BDDs: $\neg,\vee,\wedge,\rightarrow$ with their standard meaning.
    
    \item  For a BDD $\mathsf{B}(\mathcal V_1,\mathcal V_2)$, the BDD $C(\mathcal V_1):=\exists \mathcal V_2 (\mathsf{B}(\mathcal V_1\mathcal ,V_2))$ satisfies the following: For $s_1\in 2^{\mathcal V_1}$,  $s_1\models \mathsf{C}(\mathcal V_1)$ if there exists $s_2\in 2^{\mathcal V_2}$ such that $(s_1,s_2)\models \mathsf{B}(\mathcal V_1,\mathcal V_2)$.
    
    \item Similarly, $C(\mathcal V_1):=\forall \mathcal V_2 (\mathsf{B}(\mathcal V_1\mathcal ,V_2))$ satisfies: For $s_1\in 2^{\mathcal V_1}$,  $s_1\models \mathsf{C}(\mathcal V_1)$ if for all $s_2\in 2^{\mathcal V_2}$, $(s_1,s_2)\models \mathsf{B}(\mathcal V_1,\mathcal V_2)$.
\end{enumerate}

We use BDDs to model sets (e.g. winning regions) in a straightforward manner i.e. $\mathsf{B}(\mathcal V)$ represent $\{s\in 2^\mathcal V : s\models \mathsf{B}(\mathcal V)\}$ .  
    We may also use BDDs to represents memoryless strategies as follows: $\mathsf{F}(\I,\O,\I',\O')$ represents the memoryless strategy by $f(s,i)=o$ iff $(s,i',o')\models\mathsf{F}(\I,\O,\I',\O')$. The paper presents symbolic constructions of BDDs. Time of such 
    symbolic algorithms is measured in \#symbolic-operations performed. 
\section{Proofs for Section~\ref{sec:algorithmsWeakBuchi}}\label{app:StratProofs}

We turn to prove Theorems~\ref{thm:weakbuchi},~\ref{thm:realiizability} and~\ref{thm:main}, given in Section~\ref{sec:algorithmsWeakBuchi}.
Throughout the appendix, we use the fact that $\omega$-regular games are {\em determined}~\cite{Martin75,2001automata}. That is, for each game with an $\omega$-regular winning condition and a state $s$, either the system has a winning strategy form $s$, or the environment has a spoiling strategy from $s$: A strategy that forces the negation of the winning condition regardless of the system's strategy.

\subsection{Proof for Theorem~\ref{thm:weakbuchi}}
\label{app:symbolicweakbuchi}

In this section we prove:

\begin{description}
\item[Theorem~\ref{thm:weakbuchi}.] \em{    Realizability and synthesis for weak \Buchi games can be done in $O(N)$ symbolic steps.}
\end{description}

We remind that we denote the union $\I\cup\O$ by $X$, and presented an $O(N)$ time construction for the next BDDs (see Section~\hyperref[par:BDD-constructions]{6.1, BDD Constructions}): 
\begin{description}
    \item[$\SCC (X,X').$] $(s,t')\models \SCC(X,X')$ iff $s$ and $t$ belong to the same SCC of $G$, the game graph of the game structure $\GS$.
    
    \item[$\terminal (X).$] $s\models \terminal (X)$ if $s$ belongs to a terminal SCC $S$; no other SCC is reachable from $S$. 
\end{description}

\begin{algorithm}[t!]

\begin{description}
    \item{\textbf{Base Case.}} $\winningAux^0(\X) := \terminal(\X) \wedge \Acc(X)$, \\
    $\C^0(\X) := \terminal(\X)$, and  \\   $\strategyAux^0(\X,\X')= \terminal(X)\wedge (\rho_\I\rightarrow\rho_\O)(X)$
    
    \item {\textbf{Inductive Step.}}
    \begin{align*}
    \C^{i+1}(\X) := & \forall \X': \Reach(\X,\X') \rightarrow (\SCC(\X,\X') \vee \C^i(X'))\\
    \C^{i+1}_{new}(\X) := & \C^{i+1}(\X) \setminus \C^i(\X) \\
    \NC^{i+1}(X):= & \C^{i+1}_{new}(X) \wedge \neg\Acc(X)\\
    \AC^{i+1}(X):= & \C^{i+1}_{new}(X) \wedge \Acc(X)
    \end{align*}
    \begin{align*}
    \winningAux^{i+1}(\X) & :=        \winningAux^i(\X) \\
    & \vee \Reachability_{(\NC^{i+1}(X), \winningAux^i(\X))}(X) \\
    &  \vee  
    \Safety_{(\AC^{i+1}(X), \AC^{i+1}(X)\vee\winningAux^i(\X))}(X)
    \end{align*}
    
         \begin{align*}
    \strategyAux^{i+1}(X,X') & := \strategyAux^i(X,X') \\
       & \vee (\NC^{i+1}(X)\wedge\FReachability_{(\NC^{i+1}(X), \winningAux^i(\X))}(X,X')) \\
    &  \vee  
    (\AC^{i+1}(X)\wedge \FSafety_{(\AC^{i+1}(X), \AC^{i+1}(X)\vee\winningAux^i(\X))}(X,X'))
    \end{align*}
  
\end{description}

    \caption{Weak \Buchi winning memoryless strategy construction \label{alg:weak-buchi} }
\end{algorithm}

To provide a rigorous proof for Theorem~\ref{thm:weakbuchi}, in Algorithm~\ref{alg:weak-buchi} we modify the weak \Buchi winning states construction from Section~\ref{sec:algorithmsWeakBuchi} to construct also a corresponding winning strategy.
We remind some notations. Given sets of states $A$ and $B$:
\begin{description}
    \item[$\Reachability_{(A,B)}(\X).$] A BDD for
the states in $A$ from which the system can force reaching $B$.

    \item[$\Safety_{(A,B)}(\X).$]
A BDD for the states in $A$ from which the system can force the play to stay in $B$.
\end{description}
In addition to those BDDs, to reason about the weak \Buchi strategy construction, we consider also their corresponding (memoryless) strategies:
$\FReachability_{(A,B)}(\X,X')$ and $ \FSafety_{(A,B)}(X,X')$.
Note that these two strategies are given as BDDs over $\I{\cup}\O{\cup}\I'{\cup}\O'$, in compliance with the discussion given in Appendix~\ref{app:prelims}.

We argue that algorithm~\ref{alg:weak-buchi} is correct.
\begin{lemma}
\label{lem:-weak-buchi-sound-and-complete}
The fixed-points of $\winningAux(X)$ and $\strategyAux(X,X')$ in Algorithm~\ref{alg:weak-buchi} are BDDs for the winning states and a winning strategy in the weak \Buchi game $(\GS,\ltlG \ltlF( \Acc))$, respectively.
\end{lemma}

\begin{proof}
Let $\winning(X)$ and $\strategy(X,X')$ be the fixed-points of $\winningAux(X)$ and $\strategyAux(X,X')$.  
We divide the proof into \textbf{soundness} and \textbf{completeness}. That is, we show that (a). $\strategy(X,X')$ wins from $\winning(X)$, and (b). no other state is winning for the system.

\textbf{Soundness.} We show, by induction, that for each $i$, $\strategyAux^i(X,X')$ is a winning strategy from $\winningAux^i(X)$. The claim clearly holds for $\strategyAux^0(X)$, because every strategy is winning from the accepting terminal SCCs.

For the induction step, take $s\in \winningAux^{i+1}(X)$:
\begin{enumerate}
    \item If $s\in \winningAux^i(X)$ we are done by the induction hypothesis (note that $\strategyAux^{i+1}(X,X')$ coincides with $\strategyAux^i(X,X')$ on $\winningAux^i(X)$). 
    
    \item f $s\in \Reachability_{(\NC^{i+1}(X), \winningAux^i(\X))}(X)$, $\strategyAux^{i+1}(X,X')$ follows the next strategy from $s$: play reachability to reach $\winningAux^i(X)$, and then apply $\strategyAux^i(X,X')$ for the rest of the play. By the induction hypothesis, any play consistent with this strategy is winning for the system. 
    
    \item Finally, if $s\in     \Safety_{(\AC^{i+1}(X), \AC^{i+1}(X)\vee\winningAux^i(\X))}(X)$, the system plays the next strategy from $s$: play safety inside $\AC^{i+1}(X)\vee\winningAux^i(\X))$, and if the play traverses to $\winningAux^i(\X)$, switch  
    to $\strategyAux^i(X,X')$. Since all states of $\AC^{i+1}(X)$ are accepting, by the induction hypothesis, any play consistent with this strategy is winning for the system.
\end{enumerate}

\textbf{Completeness.} We show, by induction, that for each $i$, the environment has a winning strategy from  $\C^i(X)\setminus\winningAux^i(X)$. The claim clearly holds for $i=0$, since every strategy for the environment  is winning from the non-accepting terminal SCCs.

For the induction step, take $s\notin \winningAux^{i+1}(X)$ (and thus $s\notin \winningAux^i(X)$).
\begin{enumerate}
    \item First, assume that $s\in \NC^{i+1}(X)$.\\ Since  $s\notin \Reachability_{(\NC^{i+1}(X), \winningAux^i(\X))}(X)$, the environment has a safety strategy on the region $\NC^{i+1}(X)\cup (\C^{i}(X)\setminus \winningAux^i(X))$ from $s$. A play consistent with this strategy either stays in $\NC^{i+1}(X)$, or reaches $\C^{i}(X)\setminus \winningAux^i(X)$. By the induction hypothesis, in both cases the environment wins.
    
    \item Otherwise, $s\in \AC^{i+1}(X)$. Since     $s\notin\Safety_{(\AC^{i+1}(X), \AC^{i+1}(X)\vee\winningAux^i(\X))}(X)$, the environment has a strategy to reach $\C^i(X)\setminus\winningAux^i(X)$ from $s$. Hence, by the induction hypothesis, the environment has a winning strategy from $s$.\qed
\end{enumerate}
\end{proof}

To complete the proof of Theorem~\ref{thm:weakbuchi}, we need to show that Algorithm~\ref{alg:weak-buchi} terminates in time $O(N)$.

\begin{lemma}
\label{lem:weak-buchi-time-analysis}
For each $i$, the computation of  $\Reachability_{(\NC^{i+1}(X), \winningAux^i(\X))}(X)$,
$\Safety_{(\AC^{i+1}(X),\AC^{i+1}(X)\vee\winningAux^i(\X)}(X)$, 
$\FReachability_{(\NC^{i+1}(X), \winningAux^i(\X))}(X,X')$, and 
$\FSafety_{(\AC^{i+1}(X),  \AC^{i+1}(X)\vee\winningAux^i(\X))}(X,X')$ is performed in $O(|\C^{i+1}(X)\setminus \C^i(X)|)$ symbolic operations.
\end{lemma}

\begin{proof}
We start with $\Reachability_{(\NC^{i+1}(X), \winningAux^i(\X))}(X)$. 
To compute this set of states we perform the next fixed-point computation in the sub-game over $\C^{i+1}(X)$: $Z_0=\winningAux^i(X)$, and $Z_{j+1}=\mathit{pre}_s(Z_j)\cup Z_j$, where $\mathit{pre}_s(Z)$ is the set of states from which the system can force reaching $Z$ in a single step. Thus, $Z_0\subseteq Z_1\subseteq \cdots$, and the winning region is $Z_j$ that satisfies $Z_j=Z_{j+1}$. We refer the reader to~\cite[Chapter~20]{2001automata} for a comprehensive discussion. 

Now, at each step $j_0< j$, at least one state is added to $Z_{j_0}$ thus there are at most $|\C^{i+1}(X)\setminus \winningAux^i(X)|$ steps, each is performed in $O(1)$ operations. Furthermore, if $s\in \C^i(X)\setminus \winningAux^i(X)$, then we never add $s$ to a set $Z_{j_0}$. This claim is argued is follows. If we add $s$ to some $Z_{j_0}$, then the system can force reaching $\winningAux^i(X)$ from $s$, in contradiction to what we proved in Lemma~\ref{lem:-weak-buchi-sound-and-complete} (specifically, the completeness).  
Hence, an upper bound of $O(|\C^{i+1}(X)\setminus\C^i(X)|)$ operations is established. 

A similar argument proves the same for $\Safety_{(\AC^{i+1}(X),\AC^{i+1}(X)\vee\winningAux^i(\X)}(X)$. This BDD is computed by  $Z_0=\AC^{i+1}(X)\vee\winningAux^i(\X)$, and $Z_{j+1}=Z_j \cap \mathit{pre}_s(Z_i)$. Thus, $Z_0\supseteq Z_1\supseteq \cdots$, and the winning region is $Z_j$ that satisfies $Z_j=Z_{j+1}$.  
As before, at each step $j_0<j$ at least a single state is removed from $Z_{j_0}$ and hence, an upper bound of $O(|\AC^{i+1}(X)\vee\winningAux^i(\X)|)$ steps holds. Furthermore, a state $s\in \winningAux^i(X)$ is never removed from every $Z_j$. This claim is argued as follows. Take $s\in \winningAux^i(X)$, and assume that $s\in Z_{j_0}\setminus Z_{j_0+1}$. 
Consequently, the system cannot win the safety game from $s$. Hence, the environment has a strategy to reach $\C^i(X)\setminus \winningAux^i(X)$ form $s$, in contradiction to what we have proved in Lemma~\ref{lem:-weak-buchi-sound-and-complete}.  Therefore, the computation terminates after at most $O( |\AC^{i+1}(X)|)\leq O(|\C^{i+1}(X)\setminus\C^i(X)|)$ steps, as required. 

Finally, by applying standard constructions, the strategies \\ $\FReachability_{(\NC^{i+1}(X), \winningAux^i(\X))}(X,X')$  and \\
$\FSafety_{(\AC^{i+1}(X),  \AC^{i+1}(X)\vee\winningAux^i(\X))}(X,X')$ are constructed in time $O(|\C^{i+1}(X)\setminus\C^i(X)|)$ as well.\qed
\end{proof}

We turn to prove Theorem~\ref{thm:weakbuchi}.
\begin{proof}[proof of Theorem~\ref{thm:weakbuchi}]
By Lemma~\ref{lem:-weak-buchi-sound-and-complete}, Algorithm~\ref{alg:weak-buchi} computes both the system's winning region and a corresponding memoryless winning strategy. Lemma~\ref{lem:weak-buchi-time-analysis} implies that the algorithm terminates in time $O(N)$.
\end{proof}

\subsection{Proofs for Theorem~\ref{thm:realiizability} and Theorem~\ref{thm:main}}
\label{appx:grk-strategy}

In this section we prove:

\begin{description}

\item[Theorem~\ref{thm:realiizability}.] {\em{Realizability for separated GR($k$) games can be reduced to realizability of weak \Buchi games.}}

\item[Theorem~\ref{thm:main}.] {\em{Let $(\mathit{\GS},\varphi)$ be a separated GR($k$) game over the input/output variables $\I$ and $\O$. Then, the realizability and synthesis problems for $(\mathit{\GS},\varphi)$ are solved in $O(|\varphi|+N)$ and $O(|\varphi|N)$ symbolic operations, respectively, where $N=|2^{\I\cup\O}|$.}}

\end{description}

Consider a Separated GR($k$) game $(\GS=(\I,\O,\theta_\I,\theta_\O,\rho_\I,\rho_\O),\varphi)$, and write  $\varphi=\bigwedge_{l=1}^k \varphi_l$, where  $\varphi_l  =
\bigwedge_{i=1}^{n_l} \ltlG \ltlF(a_{l,i}) \rightarrow \bigwedge_{j=1}^{m_l} \ltlG \ltlF(g_{l,j})$. Let $\acc$ be the assertion expressing the set of accepting states, as depicted in Section~\ref{sec:highlevel}:
For a state $s$ in an SCC $S$, with projections  $S_\I$ on $G_\I$ and $S_\O$ on $G_\O$, $s$ is accepting
 if for \textit{every} constraint $\varphi_l$, where $l\in\{1,\dots,k\}$, one of the following holds:
\begin{description}
\item[All guarantees hold in $S$.] For every $j\in\{1,\dots,m_l\}$, there exists $\out\in S_\O$ such that $\out\models g_{l,j}$.

\item[Some assumption cannot hold in $S$.] There exists $j\in\{1,\dots,n_l\}$ such that for all $\inp\in S_\I$, $\inp\not\models a_{l,j}$. 
\end{description}

\begin{lemma} \label{lem:acc}
A BDD $\Acc (X)$ such that $s\models\Acc(X)$ iff $s$ is accepting, can be constructed in $O(|\varphi|+N)$ symbolic operations.
\end{lemma}
\begin{proof}
Consider the BDD $\SCC(X,X')$, constructed in $O(N)$ operations (see Section~\ref{app:symbolicweakbuchi}), and take a component $\varphi_l= 
\bigwedge_{i=1}^{n_l} \ltlG\ltlF(a_{l,i}) \rightarrow \bigwedge_{j=1}^{m_l} \ltlG\ltlF(g_{l,j})$. 
First, we identify all states that their SCC includes a $g_{l,j}$-state. For $j\in\{1,\dots,m_l\}$, let $\mathsf{GAR}_{_{l,j}}(X,X')$ be a BDD such that $(i_0,o_0,i_1',o_1')\models \mathsf{GAR}_{{l,j}}(X,X')$ iff $o_1\models g_{l,j} $.
Write $\mathsf{SGAR}_{l,j}(X):= \exists X'(\SCC \wedge \mathsf{GAR}_{{l,j}})$. Therefore, $s\models \mathsf{SGAR}_{l,j}(X)$ iff the SCC of $s$ includes a $g_{l,j}$-state.

Now, we focus on the assumptions. In similar to the former case, let $\mathsf{ASM}_{{l,j}}(X,X')$ be a BDD  such that $(i_0,o_0,i_1',o_1')\models \mathsf{ASM}_{{l,j}}$ iff $\inp_1\models a_{l,i}$.
Write 
 $\mathsf{SASM}_{l,j(X)}:=\forall X'(\SCC(X,X')\rightarrow \neg\mathsf{ASM}_{{l,j}}(X,X'))$. Therefore, $s\models \mathsf{SASM}_{l,j}(X)$ iff $s$'s SCC does not include an $a_{l,j}$-state.

As a result, the BDD  $\Acc:=\bigwedge_{l=1}^k(  (\bigwedge_{j=1}^{m_l}\mathsf{SGAR}_{l,j}) \vee (\bigvee_{i=1}^{n_l} \mathsf{SASM}_{l,i}))$ satisfies the require, and its construction is performed in time $O(|\varphi|+N)$. \qed
\end{proof}

  \begin{algorithm}[t!]
  	\caption{A separated GR($k$) Winning Strategy}
  	\label{algo:winning-strategy}
  	\begin{algorithmic}[1]
		
		
		\Statex{\begin{center}
  				/*\quad\quad  Implementation of $f(s, i)$ 
  				\quad\quad*/
 		\end{center}}
	
	
	    
	    \If{$s\models \acc$}
	    
	        \State \Return $f_\travel(s,\inp)$
	    
	    \Else
	    
	        \State $\mathit{mem}=0$ \label{line:reset}
	    
	        \State \Return $f_B(s,\inp)$
	    
	    \EndIf
	    
		
  	\end{algorithmic}

  	\begin{algorithmic}[1]
		
		\Statex{\begin{center}
  				/*\quad\quad  Implementation of $f_\travel(s,i)$ \quad\quad*/
 		\end{center}}
	\setcounter{ALG@line}{6}
	
  	    \For{$i=0$ to $m-1$}
	    
	    \State $\mathit{nxt\_gar}=\mathit{mem}+i \mod m$
	    
  	    \State $\out= f_{r(\mathit{nxt\_gar})}(s,i)$
  	    \If{$\out\neq \bot$}
  	        \If{$(i, o) \models \gar_{\mathit{nxt\_gar}}$}
  	            \State $\mathit{mem} = \mathit{nxt\_gar} + 1$ \label{line:increment}
  	        \EndIf
  	        \State \Return $\out$ \label{line:return-reach}
  	    \EndIf
  		\EndFor
  		\State \Return $\out$ in the same SCC $S_\O$ of $s|_\O$
		\label{line:return-arbitrary}

  	\end{algorithmic}

  \end{algorithm}

After constructing the BDD $\Acc(\X)$ for $\acc$, we can obtain a winning strategy $f_B$ for the weak \Buchi game $(\GS, \ltlG \ltlF (\acc))$ via Algorithm~\ref{alg:weak-buchi}. Algorithm~\ref{algo:winning-strategy} then describes how to construct a winning strategy $f$ for the separated GR($k$) game $(\GS, \varphi)$ using $\acc$, $f_B$ and the strategy $f_\travel$ for visiting all guarantees in an SCC.

The strategy $f_\travel$ is described in Section~\ref{sec:symsol} and detailed in the second part of Algorithm~\ref{algo:winning-strategy}. Note that, for the sake of presentation, we assume that if  $s\in S$, an SCC that includes no $\gar_j$-state, then the reachability strategy $f_{r(j)}(s, \inp)$ for guarantee $\gar_j$ returns a distinguished value $\bot$. If all $f_{r(j)}$ return $\bot$, then no guarantee is satisfiable in $S$, and therefore $f_\travel$ simply returns an arbitrary output that remains inside $S$ (line~\ref{line:return-arbitrary}). Also note that $f_\travel$ updates its memory variable $\mathit{mem}$ whenever it reaches a state that satisfies a guarantee (line~\ref{line:increment}), and $f$ resets $\mathit{mem}$ whenever it switches from $f_\travel$ to $f_B$ (line~\ref{line:reset}).

\begin{lemma}
\label{lem:gr(k)-strategy-soundeness}
Let $W$ be the system's winning region for the weak \Buchi game $B = (\GS, \ltlG \ltlF (\acc))$.  Then, Algorithm~\ref{algo:winning-strategy} constructs a winning strategy from $W$.
\end{lemma}
\begin{proof}
We will prove that if $s \in W$, then the strategy $f$ described in Algorithm~\ref{algo:winning-strategy} wins $(\GS, \varphi)$ from $s$. We reason by cases.

First, consider the case where $s \models acc$, i.e., $s$ is in an accepting SCC $S$. Then, $f(s, i)$ returns the output of $f_\travel(s, i)$. There are two possibilities:

\begin{itemize}
    \item $f_\travel$ returns from line~\ref{line:return-arbitrary}. This is only possible if $f_{r(\mathit{nxt\_gar})}(s, i) = \bot$ for every $\mathit{nxt\_gar} \in \{0, \ldots, m - 1\}$, meaning that there is no state in $S$ that satisfies any of the guarantees $\gar_\mathit{nxt\_gar}$. But by definition of accepting SCCs, this means that there must be no $\varphi_l$ for which all assumptions can be satisfied from inside $S$. Therefore, as long as the play remains inside $S$, $\varphi$ is vacuously satisfied and the system wins. Note that in this case, $f_\travel$ will always return from line~\ref{line:return-arbitrary}, and therefore always remain in $S$ unless the environment chooses an input that takes the play to a different SCC. In that case, however, by Theorem~\ref{thm:delay} the new SCC will also be inside the winning region.
    
    \item $f_\travel$ returns from line~\ref{line:return-reach}, and therefore returns the output from some function $f_{r(\mathit{nxt\_gar})}$. By construction of $f_{r(\mathit{nxt\_gar})}$, the output remains inside the same SCC and makes progress toward reaching a state that satisfies $\gar_\mathit{nxt\_gar}$. Note that the same $f_{r(\mathit{nxt\_gar})}$ will be chosen next time until $\gar_\mathit{nxt\_gar}$ is reached, at which point $\mathit{mem}$ will be updated to $\mathit{nxt\_gar} + 1\mod m$. Therefore, as long as the play remains inside $S$, $f_\travel$ will cycle through all $\gar_\mathit{nxt\_gar}$ that can be satisfied from $S$. This means that every $\varphi_l$ such that all guarantees hold in $S$ will be satisfied. Since $S$ is an accepting SCC, the only remaining $\varphi_l$s are those where some assumption cannot hold in $S$, which will therefore be vacuously satisfied. Therefore, the system is guaranteed to win while it remains in $S$. If the environment chooses an input that takes the play to a different SCC, again by Theorem~\ref{thm:delay} the new SCC will also be inside the winning region.
\end{itemize}

Next, consider the case where $s \not\models acc$, i.e., $s$ is in a non-accepting SCC. In this case, $f(s, i)$ sets $\mathit{mem}$ to $0$ and returns the output of $f_B(s, i)$. Since $s \in W$ and $f_B$ is a winning strategy for $B$, $f_B(s, i)$ is guaranteed to remain in $W$ and move towards an accepting SCC, as this is the only way for the system to win $B$. Since $f(s, i)$ will keep returning $f_B(s, i)$ while in a non-accepting SCC, the play is guaranteed to eventually reach an accepting SCC $S$. At this point $f$ will switch back to $f_\travel$.

Therefore, $f$ is guaranteed to always satisfy the winning condition $\varphi$ while the play remains inside an accepting SCC (via $f_\travel$), and it is guaranteed to always reach an accepting SCC from a non-accepting SCC (via $f_B$). Since the set of SCCs is finite and partially ordered by reachability, every play is guaranteed to eventually stay forever in a single SCC $S$. Then, the two conditions above guarantee that $S$ will be accepting and that $f$ will win from $S$. \qed
\end{proof}

\begin{lemma}
\label{lem:gr(k)-strategy-completness}
Let $W$ be the system's winning region for the weak \Buchi game $B = (\GS, \ltlG \ltlF (\acc))$.  Then, if $s\notin W$, the system cannot win the separated GR($k$) game $(\GS,\varphi)$ from $s$.
\end{lemma}
\begin{proof}
Since $s \not\in W$, there is no strategy that allows the system to force the play to converge to an accepting SCC. Since every play must converge to some SCC, the environment can force the play to converge to a non-accepting SCC from $s$.

Note that an SCC $S$ is non-accepting if there is at least one $\varphi_l$ that the system cannot satisfy from inside $S$. That is, the environment can visit all the assumptions but the system cannot visit all the guarantees in $S$.

Therefore, we can construct a function $f^e_\travel$ for the environment that is analogous to the function $f_\travel$ for the system but cycles through all assumptions rather than all guarantees. Note that $f^e_\travel$ is guaranteed to win for the environment as long as the play remains inside a non-accepting SCC.

Finally, note that a Delay Property analogous to Theorem~\ref{thm:delay} can be proved for the environment by simply swapping inputs and outputs. Therefore, the environment can ensure that the play always remains outside of $W$, it converges to a non-accepting SCC, and the environment can win from that SCC (via $f^e_\travel$). Therefore, the system cannot win the separated GR($k$) game from $s$.\qed
\end{proof}

Clearly, Theorem~\ref{thm:realiizability} immediately follows from Lemmas~\ref{lem:gr(k)-strategy-soundeness} and~\ref{lem:gr(k)-strategy-completness}. For Theorem~\ref{thm:main}, the complexity of realizability follows from Lemmas~\ref{lem:acc},~\ref{lem:gr(k)-strategy-soundeness} and~\ref{lem:gr(k)-strategy-completness} and Theorem~\ref{thm:weakbuchi}. The complexity of synthesis follows from realizability with the addition of the computation of $f_\travel$, which requires computing a reachability game for each $f_{r(j)}$, taking $O((\#\mathit{gars})N)\leq O(|\varphi|N)$ symbolic operations. 






\end{document}